\documentclass[11pt,letterpaper]{article}
\usepackage[margin=1in]{geometry}
\usepackage{fullpage}
\usepackage{url}
\usepackage{amsmath}
\usepackage{amsfonts}
\usepackage{authblk}
\usepackage{hyperref}

\usepackage{xcolor}

\usepackage{graphicx}
\graphicspath{ {FIGURES/} }
\usepackage{caption}
\usepackage{subcaption}

\newtheorem{theorem}{Theorem}
\newtheorem{lemma}{Lemma}
\newtheorem{remark}{Remark}
\newenvironment{proof}{\noindent{\bf Proof:}}{\hfill$\Box$}

\def\cvN{c_v^{(N)}}
\def\ceN{c_e^{(N)}}
\def\ReN{R_e^{(N)}}
\def\oeN{o_e^{(N)}}

\def\ol{\widehat}
\def\max{\mathrm{max}}

\def\cE{\mathcal{E}}
\def\cF{\mathcal{F}}
\def\cH{\mathcal{H}}

\title{Phase transition in count approximation by Count-Min sketch with conservative updates}

\author[1]{\'Eric Fusy}
\author[1]{Gregory Kucherov}
\affil[1]{LIGM, CNRS, Univ. Gustave Eiffel, Marne-la-Vall\'ee, France \texttt{\{Eric.Fusy|Gregory.Kucherov\}@univ-eiffel.fr}}

\date{}

\begin{document}
	\maketitle
	
	\begin{abstract}
		Count-Min sketch is a hash-based data structure to represent a dynamically changing associative array of counters. Here we analyse the counting version of Count-Min under a stronger update rule known as \textit{conservative update}, assuming the uniform distribution of input keys. We show that the accuracy of conservative update strategy undergoes a phase transition, depending on the number of distinct keys in the input as a fraction of the size of the Count-Min array. We prove that below the threshold, the relative error is asymptotically $o(1)$ (as opposed to the regular Count-Min strategy), whereas above the threshold, the relative error is $\Theta(1)$. The threshold corresponds to the peelability threshold of random $k$-uniform hypergraphs. We demonstrate that even for small number of keys, peelability of the underlying hypergraph is a crucial property to ensure the $o(1)$ error. Finally, we provide an experimental evidence that the phase transition does not extend to non-uniform distributions, in particular to the popular Zipf's distribution. 
	\end{abstract}

	\section{Introduction}
	\textit{Count-Min sketch} is a hash-based data structure to represent a dynamically changing associative array $\boldsymbol{a}$ of counters in an approximate way. The array $\boldsymbol{a}$ can be seen as a mapping from some set $K$ of keys to $\mathbb{N}$, where $K$ is drawn from a (large) universe $U$.  The goal is to support {\em point queries} about the (approximate) current value of $\boldsymbol{a}(p)$ for a key $p$. Count-Min is especially suitable for the streaming framework, when counters associated to keys are updated dynamically. That is, {\em updates} are (key,value) pairs $(p,\ell)$ with the meaning that $\boldsymbol{a}(p)$ is updated to $\boldsymbol{a}(p)+\ell$. 
	
	Count-Min sketch was proposed in \cite{cormode_improved_2005}, see e.g. \cite{DBLP:reference/db/Cormode18} for a survey. A similar data structure was introduced earlier in \cite{DBLP:conf/sigmod/CohenM03} named \textit{Spectral Bloom filter}, itself closely related to \textit{Counting Bloom filters} \cite{FanEtAl00}. The difference between Count-Min sketch and Spectral Bloom filter is marginal: while a Count-Min sketch requires hash functions to have disjoint codomains (rows of Count-Min matrix), a Spectral Bloom filter has all hash functions mapping to the same array, as does the regular Bloom filter. In this paper, we will deal with the Spectral Bloom filter version but will keep the term Count-Min sketch as more common in the literature. 
	
	Count-Min sketch supports negative 	update values $\ell$ provided that at each moment, each counter $\boldsymbol{a}(p)$ remains non-negative (so-called {\em strict turnstile model} \cite{liu2011methods}). 
	When updates are positive, the Count-Min update algorithm can be modified to a stronger version leading to smaller errors in queries. This modification, introduced in \cite{DBLP:conf/sigcomm/EstanV02} as \textit{conservative update}, is mentioned \cite{DBLP:reference/db/Cormode18}, without any formal analysis given in those papers. This variant is also discussed in \cite{DBLP:conf/sigmod/CohenM03} under the name \textit{minimal increase}, where it is claimed that it decreases the probability of a positive error by a factor of the number of hash functions, but no proof is given. We discuss this claim in the concluding part of this paper. 
	
	\newcommand{\lone}{\lVert \boldsymbol{a} \rVert_{1}}
	The case of positive updates is widespread in practice. In particular, a very common instance is \textit{counting} where all update values are $1$. 	This task occurs in different scenarios in network traffic monitoring, as well as other applications related to data stream mining \cite{DBLP:conf/sigcomm/EstanV02}. In bioinformatics, we may want to maintain, on the fly, multiplicities of $k$-mers (words of length $k$) occurring in a big dataset \cite{mohamadi2017ntcard,behera2018kmerestimate,shibuya_set-min_2020}. We refer to \cite{DBLP:conf/iccnc/EinzigerF15} for more examples of applications. 
	
	While it is easily seen that the error in conservative update can only be smaller than in Count-Min, obtaining more precise bounds is a challenging problem. Count-Min guarantees, with high probability, that the additive error can be bounded by $\varepsilon \lVert \boldsymbol{a} \rVert_{1}$ for any $\varepsilon$, where $\lVert \boldsymbol{a} \rVert_{1}$ is the $L1$-norm of $\boldsymbol{a}$ \cite{cormode_improved_2005}. In the counting setting, $\lVert \boldsymbol{a} \rVert_{1}$ is the length of the input stream which can be very large, and therefore this bound provides a weak guarantee in practice, unless the distribution of keys is very skewed and queries are made on frequent keys (\textit{heavy hitters}) \cite{liu2011methods,charikar2004finding,cormode2008finding}. 
	It is therefore an important practical question to analyse the improvement provided by the conservative update strategy compared to the original Count-Min sketch. 
	
	Probably the first attempt towards this goal was made in \cite{DBLP:conf/teletraffic/BianchiDLS12}, under assumption that all $\binom{n}{k}$ counter combinations are equally likely at each step ($n$ size of the Count-Min array, $k$ number of hash functions) which amounts to assuming uniform distribution on $\binom{n}{k}$ input keys. {
		For the regime when the number of distinct keys in the input considerably exceeds the sketch size, it was experimentally observed in \cite{DBLP:conf/teletraffic/BianchiDLS12} that frequent keys have essentially no error  while the (over)estimate for less frequent keys tends be the same, for given input stream length and sketch size. 
		On the other hand, 
		it was experimentally shown
		that in that regime, the uniform distribution of keys presents the worst-case scenario producing the largest 
		such estimate and a method was proposed to compute this estimate. 
		Another method for bounding the error proposed in \cite{DBLP:conf/iccnc/EinzigerF15} 
	}     
	is based on a simulation of spectral Bloom filters  by a hierarchy of ordinary Bloom filters. 
	However, the bounds provided are not explicit but are expressed via a recursive relation based on false positive rates of involved Bloom filters. 
Recent works \cite{DBLP:journals/corr/abs-2203-14549,benmazziane:hal-03613957} propose formulas for computing error bounds depending on key probabilities assumed independent but not necessarily uniform, in particular leading to an improved precision bounds for detecting heavy hitters. 
	
	In this paper, we provide a probabilistic analysis of the conservative update scheme for counting under the assumption of \textit{uniform distribution of keys} in the input. Our main result is a demonstration that the error in count estimates undergoes a phase transition 
	when the number of distinct keys grows relative to the size of the Count-Min array. We show that the phase transition threshold corresponds to the \textit{peelability threshold} for random $k$-uniform hypergraphs. 
	For the \textit{subcritical regime}, when the number of distinct keys is below the threshold, we show that the relative error for a randomly chosen key tends to $0$ asymptotically, with high probability. This contrasts with the regular Count-Min algorithm producing a relative error shown to be at least $1$ with constant probability. 
	
	For the \textit{supercritical regime}, we show that the average relative error is lower-bounded by a constant (depending on the number of distinct keys), with high probability. We prove this result for $k=2$ and conjecture that it holds for arbitrary $k$ as well. We provide computer simulations showing that the expected relative error grows fast after the threshold, with a distribution showing a peculiar multi-modal shape. In particular, keys with small (or zero) error still occur after the threshold, but their fraction quickly decreases when the number of distinct keys grows. 
	
	
	After defining Count-Min sketch and conservative update strategy in Section~\ref{countmin-and-CU} and introducing hash hypergraphs in Section~\ref{sec:hash}, we formulate the conservative update algorithm (or regular Count-Min, for that matter) in terms of a hypergraph augmented with counters associated to vertices. In Section~\ref{sec:main}, we state our main results and illustrate them with a series of computer simulations. 
	All technical proofs are provided in a separate Section~\ref{sec:proofs}. 
		
	In addition, in Section~\ref{sec:non-peelable}, we study a specific family of 2-regular $k$-hypergraphs that are sparse but not peelable.  For such graphs we show that while the relative error of every key is $1$ with the regular Count-Min strategy, it is $1/k+o(1)$ for conservative update. While this result is mainly of theoretical interest, it illustrates that the peelability property is crucial for the error to be asymptotically vanishing. 
	Finally, in Section~\ref{sec:zipf}, we turn to non-uniform distributions and provide experimental evidence that for Zipf's distribution, the phase transition in average error does not occur. 
	
\section{Count-Min and Conservative Update}
\label{countmin-and-CU}
We consider a (counting version of) Count-Min sketch to be an array $A$ of size $n$ of counters initially set to $0$, together with hash functions $h_1,\ldots,h_k$ mapping keys from a given universe to $[1..n]$. 
To count key occurrences in a stream of keys, regular Count-Min proceeds as follows. 
To process a key $p$, each of the counters $A[h_i(p)]$, $1\leq i\leq k$, is incremented by $1$. 
Querying the occurrence number $\boldsymbol{a}(p)$ of a key $p$ returns the estimate $\hat{\boldsymbol{a}}_{CM}(p)=\min_{1\leq i\leq k}\{A[h_i(p)]\}$. It is easily seen that $\hat{\boldsymbol{a}}_{CM}(p)\geq \boldsymbol{a}(p)$. 
A bound on the overestimate of $\boldsymbol{a}(p)$ is given by the following result adapted from \cite{cormode_improved_2005}.
	\begin{theorem}[\cite{cormode_improved_2005}]
		\label{countmin}
		For $\varepsilon>0$, $\delta>0$, consider a Count-Min sketch with $k=\lceil \ln(\frac{1}{\delta}) \rceil$  and size $n=k\frac{e}{\varepsilon}$. Then $\hat{\boldsymbol{a}}_{CM}(p)- \boldsymbol{a}(p)\leq \varepsilon N$ with probability at least $1-\delta$, where $N$ is the size of the input stream. 
	\end{theorem}
While Theorem~\ref{countmin} is useful in some situations, 
it has a limited utility as it bounds the error with respect to the stream size which can be very large. 

\textit{Conservative update} strengthens Count-Min by increasing only the smallest counters among $A[h_i(p)]$. Formally, for $1\leq i\leq k$, 
$A[h_i(p)]$ is incremented by $1$ only if $A[h_i(p)]=\min_{1\leq j\leq k}\{A[h_j(p)]\}$ and is left unchanged otherwise. 
The estimate of $\boldsymbol{a}(p)$, denoted $\hat{\boldsymbol{a}}_{CU}(p)$, is computed as before: $\hat{\boldsymbol{a}}_{CU}(p)=\min_{1\leq i\leq k}\{A[h_i(p)]\}$. 
It can be seen that $\hat{\boldsymbol{a}}_{CU}(p)\geq \boldsymbol{a}(p)$ still holds, and that $\hat{\boldsymbol{a}}_{CU}(p)\leq \hat{\boldsymbol{a}}_{CM}(p)$. 
The latter follows from the observation that on the same input, an entry of counter array $A$ under conservative update can never get larger than the same entry under Count-Min. 


\section{Hash hypergraphs and CU process}\label{sec:hash}

With a counter array $A[1..n]$ and hash functions $h_1,...,h_k$ we associate a $k$-uniform \textit{hash hypergraph}  $H=(V,E)$ with vertex-set $V=\{1..n\}$
 and edge-set $E=\{\{h_1(p),...h_k(p)\}\}$ for all distinct keys $p$. Let $\cH_{n,m}^{k}$ be the set of $k$-uniform hypergraphs with $n$ vertices and $m$ edges. 
%
We assume that the hash hypergraph is a uniformly random Erd\H{o}s-R\'enyi hypergraph in $\cH_{n,m}^{k}$, which 
we denote by $H^k_{n,m}$, where $m$ is the number of distinct keys in the input (for $k=2$, we use the notation $G_{n,m}=H^2_{n,m}$). 
Even if this property is not granted by hash functions used in practice, it is a reasonable and commonly used hypothesis to conduct the analysis of sketch algorithms.   

Below we show that the behavior of a sketching scheme depends on the properties of the associated hash hypergraph. It is well-known that depending on the $m/n$ ratio, many properties of Erd\H{o}s-R\'enyi (hyper)graphs follow a phase transition phenomenon \cite{FriezeKaronski16}. For example, the emergence of a giant component, of size $O(n)$, occurs {with high probability} (hereafter, \textit{w.h.p.}) at the threshold $\frac{m}{n}=\frac{1}{k(k-1)}$ \cite{KARONSKI2002125}. 


Particularly relevant to us is the \textit{peelability} property. 
Let $H=(V,E)$ be a hypergraph. The peeling process on $H$ is as follows. We define $H_0=H$, and iteratively for $i\geq 0$, we define $V_i$  to be the set of leaves (vertices of degree 1) or isolated vertices in $H_i$, $E_i$ to be the set of edges of $H_i$ incident to vertices in $V_i$, and $H_{i+1}$ to be the hypergraph obtained from $H_i$ by deleting  
the vertices of $V_i$ and the edges of $E_i$. A vertex in $V_i$ is said to have \emph{peeling level}~$i$. 
The process stabilizes from some step $I$, and the hypergraph $H_I$ is called the 
\emph{core} of $H$, which is the largest induced sub-hypergraph whose vertices all have degree at least $2$. If $H_I$ is empty, then $H$ is called \emph{peelable}.

It is known~\cite{molloy2005cores} that peelability undergoes a phase transition. For $k\geq 3$, there exists a positive constant $\lambda_k$ such that, for $\lambda<\lambda_k$,
the random hypergraph $H^k_{n,\lambda n}$ is w.h.p. peelable as $n\to\infty$, while for $\lambda>\lambda_k$, the core of $H^k_{n,\lambda n}$ has w.h.p. a size  concentrated around 
$\alpha n$ for some $\alpha>0$ that depends on $\lambda$.  The first peelability thresholds are $\lambda_3\approx 0.818$, $\lambda_4\approx 0.772$, etc., $\lambda_3$ being the largest. 

For $k=2$, for $\lambda<1/2$, w.h.p. a proportion $1-o(1)$ of vertices are in trees of size $O(1)$, 
(and a proportion $o(1)$ of the vertices are in the core), while for $\lambda\geq 1/2$, the core size is w.h.p. concentrated around $\alpha n$ for $\alpha>0$ that depends on~$\lambda$
~\cite{pittel2005counting}. 

We note that properties of hash hypergraphs determine the behavior of some other hash-based data structures, such as Cuckoo hash tables \cite{pagh2004cuckoo} and Cuckoo filters \cite{10.1145/2674005.2674994}, Minimal Perfect Hash Functions and Static Functions \cite{majewski1996family}, Invertible Bloom filters \cite{goodrich2011invertible}, and others. We refer to \cite{DBLP:phd/dnb/Walzer20} for an extended study of relationships between properties of hash hypergraphs and some of those data structures. In particular, peelability is directly relevant to certain constructions of Minimal Perfect Hash Functions as well as to good functioning of Invertible Bloom filters.

%


The connection to hash hypergraphs allows us to reformulate the Count-Min algorithm with conservative updates as a process, which 
we call CU-process, on a random hypergraph $H_{n,m}^k$, where $n,m,k$ correspond to counter array length, number of distinct keys, and number of 
hash functions, respectively. 
Let $H=(V,E)$ be a hypergraph. To each vertex $v$ we associate a counter $c_v$ initially set to $0$.  At each step $t\geq 1$, a \emph{CU-process} on $H$ chooses  an edge $e=\{v_1,\ldots,v_k\}\in E$ in $H$, and increments by 1 those $c_{v_i}$ which verify $c_{v_i}=\min_{1\leq j\leq k}c_{v_j}$. 
For $t\geq 0$ and $v\in V$, $c_v(t)$ will denote the value of the counter $c_v$ after $t$ steps, and $o_e(t)$ the number of times edge $e\in E$ has been drawn in the first $t$ steps. The counter $c_e(t)$ of an edge $e=\{v_1,\ldots,v_k\}$ is defined as $c_e(t)=\min_{1\leq i\leq k} c_{v_i}(t)$. Clearly, for each $t$ and each $e$, $o_e(t)\leq c_e(t)$. The \emph{relative error} of $e$ at time $t$ is defined as $R_e(t)=\frac{c_e(t)-o_e(t)}{o_e(t)}$. 
The following Lemma can be easily proved by induction on $t$.
\begin{lemma}\label{lem:observation}
	Let $H=(V,E)$ be a hypergraph on which a CU-process is run. 
	At every step $t$, for each vertex $v$, there is at least one edge $e$ incident to $v$ such that $c_e(t)=c_v(t)$.  
\end{lemma}
Observe that, when $H$ is a graph ($k=2$), Lemma~\ref{lem:observation} is equivalent to the property that vertex counters cannot have a strict local maximum, i.e., at every step $t$, each vertex $v$ has at least one neighbour~$u$
such that $c_u(t)\geq c_v(t)$. 

 \section{Phase transition of the relative error}
 \label{sec:main}

\subsection{Main results}\label{sec:main_results}

Let $H=(V,E)$ be a hypergraph, $|V|=n, |E|=m$. Let $N\geq 1$. We consider two closely related models of input to perform the CU-process. In the \emph{$N$-uniform model}, the CU process is performed on a random sequence of keys (edges in $E$)  
of length $N\cdot m$, each key being drawn independently and uniformly in $E$. 
In  the \emph{$N$-balanced model}, the CU-process is performed on a random sequence of length $N\cdot m$,  
such that each $e\in E$ occurs exactly $N$ times, and the order of keys is random. In other words, the sequence of keys on which the 
CU-process is performed is a random permutation of the multiset made of $N$ copies of each key of $E$. 
Clearly, both models are very close, since the number of occurrences of any key in the $N$-uniform model is concentrated around $N$ (with Gaussian fluctuations 
of order $\sqrt{N}$) as $N$ gets large.   
For both models, we use the notation $\cvN=c_v(Nm)$ for the resulting counter of $v\in V$, $\oeN=o_e(Nm)$ for the resulting number of occurrences of $e\in E$,  
$\ceN=c_e(Nm)$ for the resulting counter of $e\in E$, and $\ReN=R_e(Nm)=(\ceN-\oeN)/\oeN$  
for the resulting relative error of $e$.  
In the $N$-balanced model, since each key $e\in E$ occurs $N$ times, we have $\ReN=(\ceN-N)/N$.  

Our main result is the following. 
 
 \begin{theorem}[subcritical regime]\label{theo:subcriti}
Let $k\geq 2$, and let $\lambda<\lambda_k$, where $\lambda_2 =1/2$, and for $k\geq 3$, $\lambda_k$ is the peelability threshold as defined in Section~\ref{sec:hash}.   
Consider a CU-process on a random hypergraph $H^k_{n,\lambda n}$ under either $N$-uniform or $N$-balanced model, and consider the relative error $\ReN$ of a random edge in $H^k_{n,\lambda n}$. Then  $\ReN=o(1)$ w.h.p., as both $n$ and $N$ grow\footnote{Formally, for any $\epsilon >0$, there exists $M$ such that $\mathbb{P}(\ReN \leq \epsilon) \geq 1-\epsilon$ if $n\geq M$ and $N\geq M$.}.
 \end{theorem}

 Note that
with the regular Count-Min algorithm (see Section~\ref{countmin-and-CU}), in the $N$-balanced model, the counter value of a node $v$ is $\tilde{c}_v^{(N)} = N\cdot \mathrm{deg}(v)$, and the relative error $\tilde{R}_e^{(N)}$ of an edge $e=(v_1,\ldots,v_k)$ is always (whatever $N\geq 1$) equal to $\mathrm{min}(\mathrm{deg}(v_1),\ldots,\mathrm{deg}(v_k))-1$, and is thus always a non-negative integer.
 For fixed  $k\geq 2$ and $\lambda>0$, and for a random edge $e$ in $H^k_{n,\lambda n}$, the probability that all $k$ vertices belonging to $e$ have at least one 
 incident edge apart from $e$ converges to a positive constant $c(\lambda,k)=(1-e^{-k\lambda})^k$. Therefore, $\tilde{R_e}$ is a nonnegative integer whose probability to be non-zero  converges to  $c(\lambda,k)$. Thus, Theorem~\ref{theo:subcriti} ensures that, for $\lambda<\lambda_k$, conservative updates lead to a drastic decrease of the error, from $\Theta(1)$ to $o(1)$.

For a given hypergraph $H=(V,E)$ with $m$ edges, we define $\mathrm{err}_N(H)=\frac1{m}\sum_{e\in E}\ReN$ the \emph{average} error over the edges of $H$. 
Theorem~\ref{theo:subcriti} states that a randomly chosen edge has a small error, but this does not formally exclude that a small fraction of edges may have large errors, possibly yielding $\mathrm{err}_N(H)$ larger than $o(1)$. However, we believe that this is not the case. From the previous remark, it follows that the error of an edge $e=(v_1,\ldots,v_k)$ is upper-bounded by $\mathrm{min}(\mathrm{deg}(v_1),\ldots,\mathrm{deg}(v_k))-1$. Since the expected maximal degree in $H_{n,\lambda n}^k$ grows very slowly with $n$, one can expect that any set of $o(n)$ edges should have a contribution $o(1)$ w.h.p.. This is also supported by experiments given in the next section. 

Based on Theorem~\ref{theo:subcriti} and the above discussion, we propound that a phase transition occurs for the average error, in the sense that it is $o(1)$ in the subcritical regime $\lambda<\lambda_k$, 
and $\Theta(1)$ in the supercritical regime $\lambda>\lambda_k$, w.h.p.. 
Regarding the supercritical regime, we are able to show that this indeed holds for $k=2$ in the $N$-balanced model.

\begin{theorem}[supercritical regime, case $k=2$]\label{theo:supercrit}
Let $\lambda>1/2$. Then there exists a positive constant $f(\lambda)$ such that, in the $N$-balanced model, $\mathrm{err}_N(G_{n,\lambda n})\geq f(\lambda)$ w.h.p., as  $n$ grows\footnote{Formally, for any $\epsilon>0$, there exists $M$ such that   $\mathbb{P}(\mathrm{err}_N(G_{n,m})\geq f(\lambda))\geq 1-\epsilon$
 if $N\geq 1$ and $n\geq M$.}. 
\end{theorem}

Our proof of Theorem~\ref{theo:supercrit} can be extended to any $k\geq 2$ and $\lambda>\frac1{k(k-1)}$ such that  
 the giant component $G'=(V',E')$  in $H_{n,\lambda n}^k$ satisfies w.h.p. $|E'|-|V'|\geq a\,n$ for some positive constant $a$. 
  For $k=2$, any value $\lambda>1/2$ has this property~\cite{pittel2005counting}, while for $k=3$ (respectively, $k=4$), the analysis given in \cite{DBLP:journals/cpc/BehrischCK14} ensures that this property holds for values of $\lambda$ strictly above the peelability threshold, namely $\lambda>\widetilde{\lambda_3}\approx 0.94$ (respectively,  $\lambda>\widetilde{\lambda_4}\approx 0.98$) . 
  %
   %
   %
  %
Nevertheless, based on simulations presented below, we expect that Theorem~\ref{theo:supercrit} holds for $k\geq 3$ for all $\lambda>\lambda_k$ as well, however proving this would then require a different kind of argument.

\subsection{Simulations}
\label{sec:simulation}

Here we provide several experimental results illustrating the phase transition stated in Theorems~\ref{theo:subcriti} and \ref{theo:supercrit}. Figure~\ref{fig:av_err} shows plots for the average relative error $\mathrm{err}_N(H_{n,m}^{k})$ as a function of $\lambda=m/n$, for $k\in\{2,3,4\}$ for regular Count-Min and the conservative update strategies. Experiments were run for $n=1000$ with the $N$-independent model (each edge drawn independently with probability $1/m$) and $N=50,000$ (number of steps $N\cdot |E|$). For each $\lambda$, an average is taken over $15$ random graphs. 

\medskip
\medskip

\begin{figure}[h]
	     \centering
	\begin{subfigure}[b]{0.35\textwidth}
		\centering
		\includegraphics[scale=0.22]{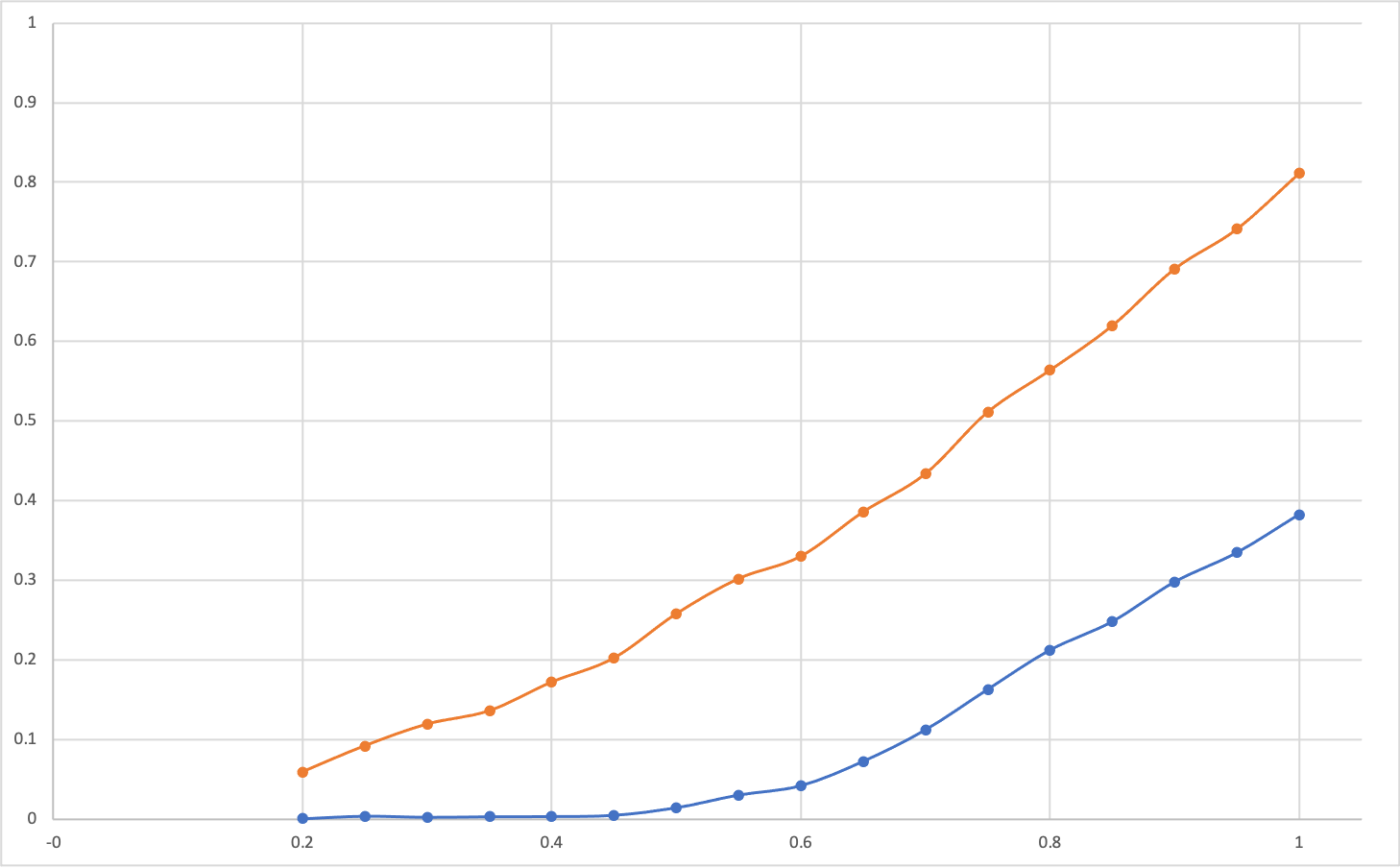}
		\caption{$k=2$}
		\label{fig:k2}
	\end{subfigure}
	\hfill
	\begin{subfigure}[b]{0.33\textwidth}
		\centering
		\includegraphics[scale=0.22]{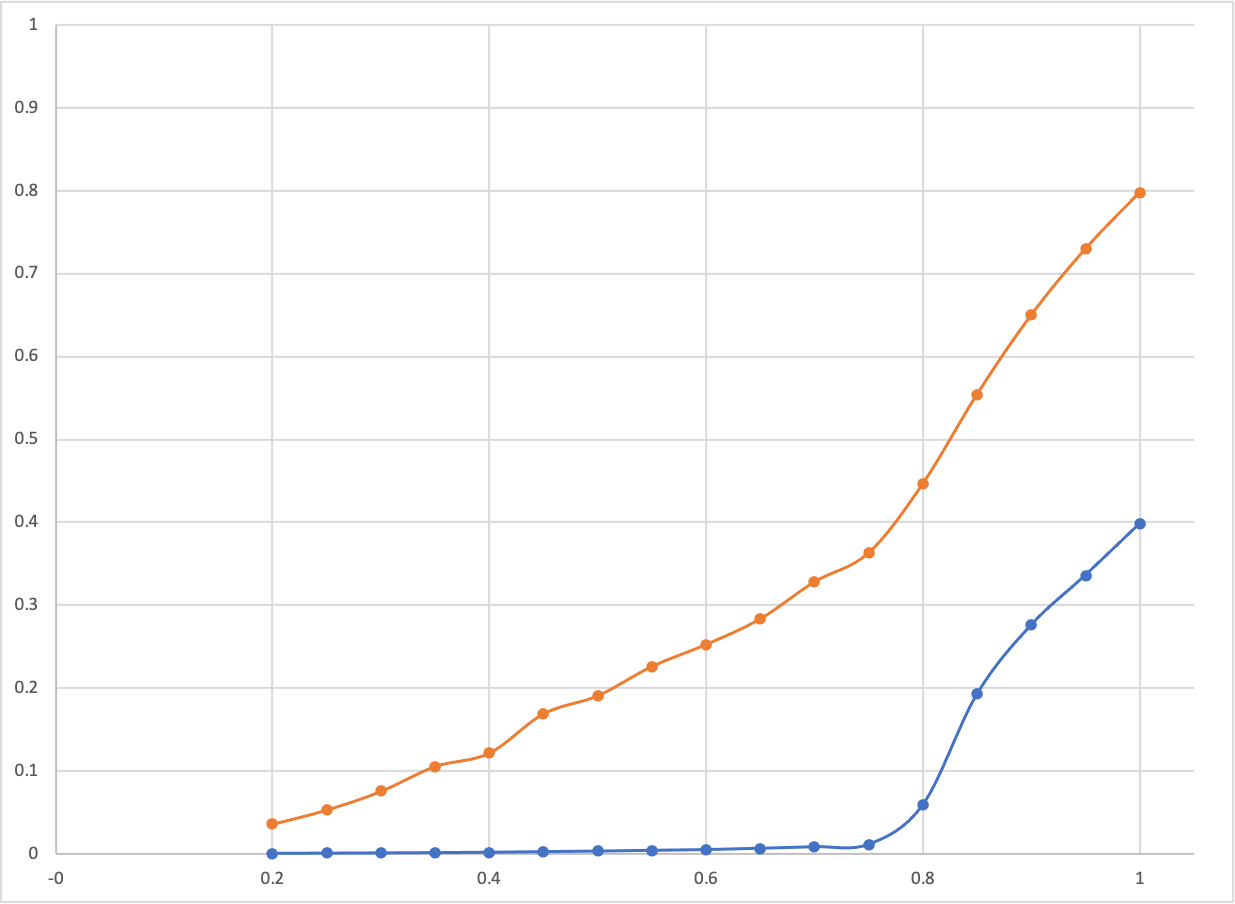}
		\caption{$k=3$}
		\label{fig:k3}
	\end{subfigure}
	\hfill
	\begin{subfigure}[b]{0.30\textwidth}
		\centering
		\includegraphics[scale=0.22]{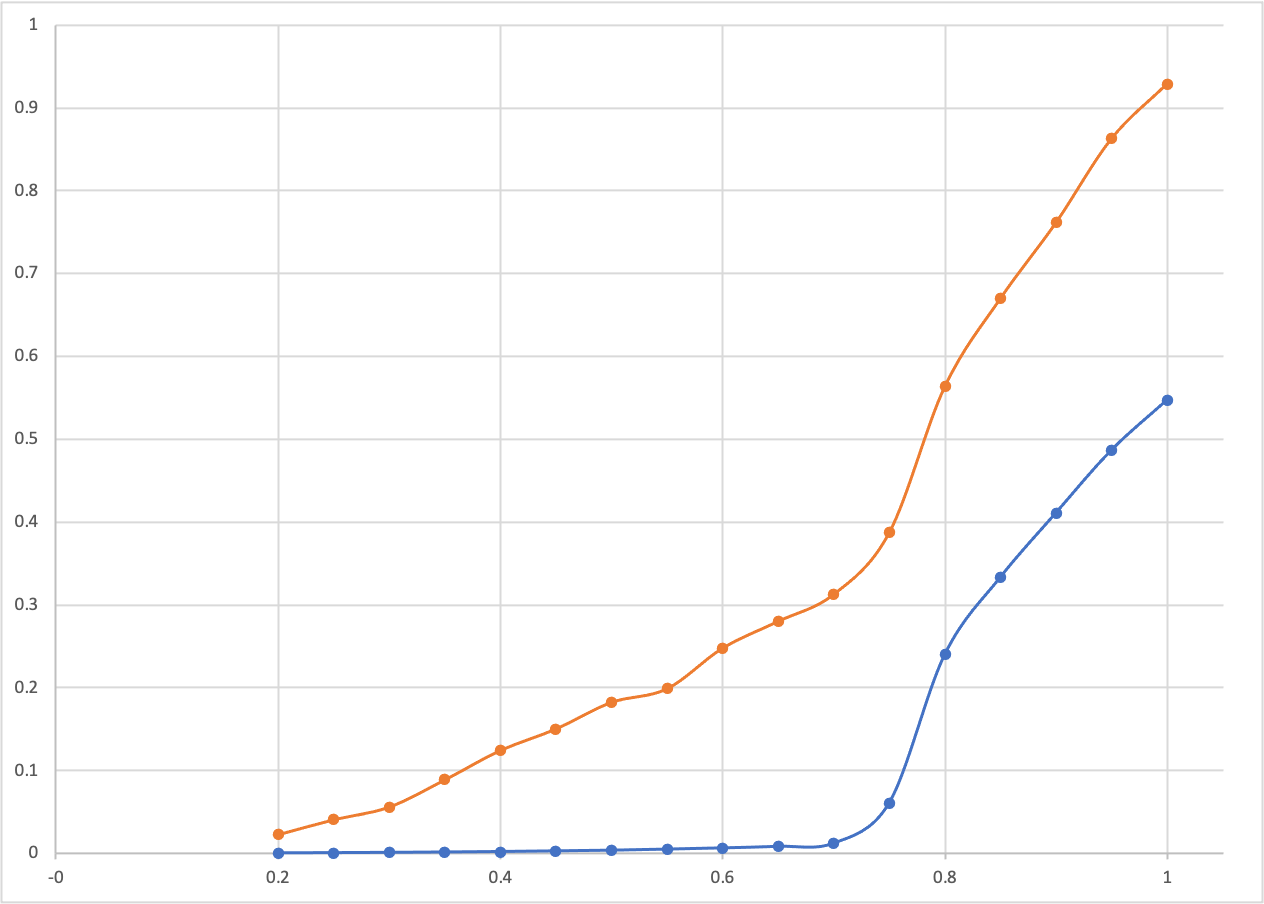}
		\caption{$k=4$}
		\label{fig:k4}
	\end{subfigure}
	\caption{Average relative error as a function of $\lambda=m/n$ for regular Count-Min (red) and conservative update (blue), for $k\in\{2,3,4\}$.}
	\label{fig:av_err}
\end{figure}

\medskip
\medskip

The phase transitions are clearly seen to correspond to the critical threshold $0.5$ for $k=2$, and, for $k\in\{3,4\}$, 
to the peelability thresholds $\lambda_3\approx 0.818$, $\lambda_4\approx 0.772$. Observe that the transition looks sharper for $k\geq 3$, which may be 
explained by the fact that the core size undergoes a discontinuous phase transition for $k\geq 3$, as shown in~\cite{molloy2005cores} (e.g. for $k=3$, the fraction of vertices in the  
core jumps from $0$ to about $0.13$). 

For the supercritical regime, we analysed the concentration around the average shown in Figure~\ref{fig:av_err} by simulating CU-processes on 50 random graphs, for each $\lambda$. Figure~\ref{fig:concentration} shows the results. We observe that the concentration tends to become tighter for growing~$\lambda$. 
\begin{figure}[h]
	\centering
	\begin{subfigure}[b]{0.47\textwidth}	
	\centering
	\includegraphics[width=\textwidth]{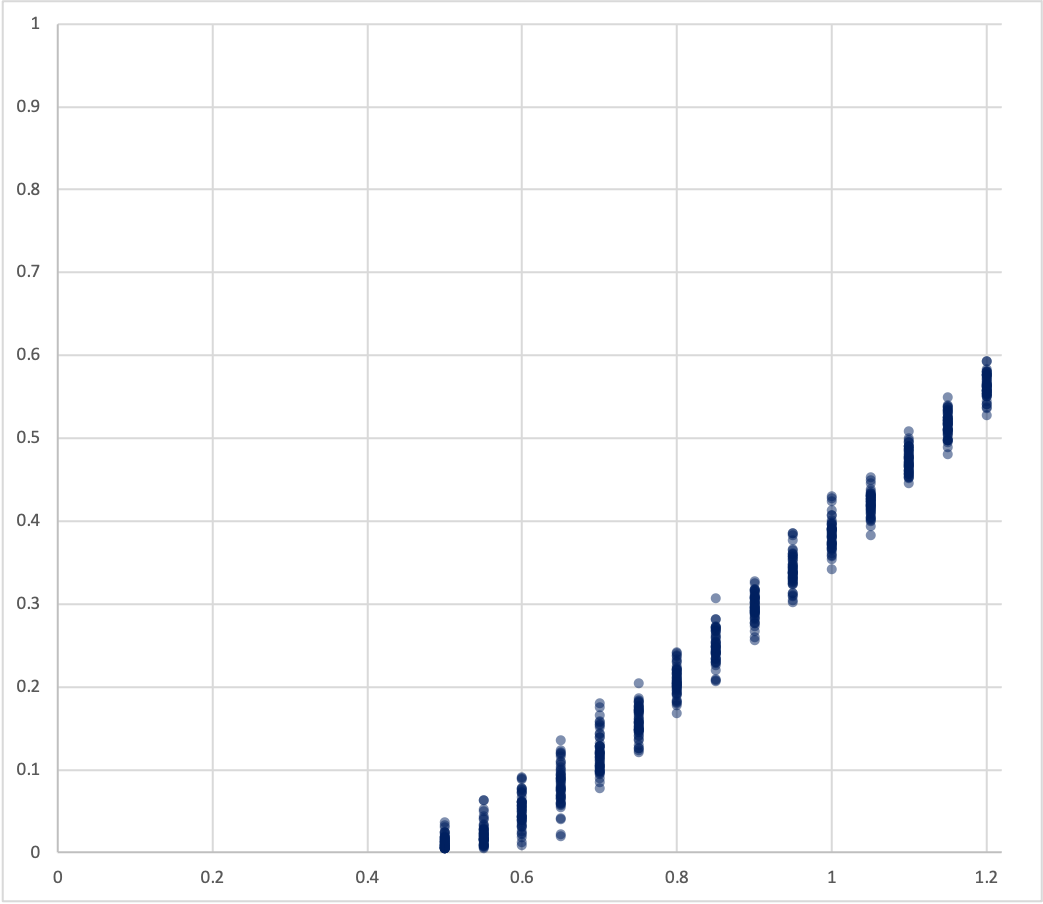}
	\caption{$k=2$}
	\label{fig:k2-concentration}
\end{subfigure}
	\hfill
\begin{subfigure}[b]{0.50\textwidth}
	\centering
	\includegraphics[width=\textwidth]{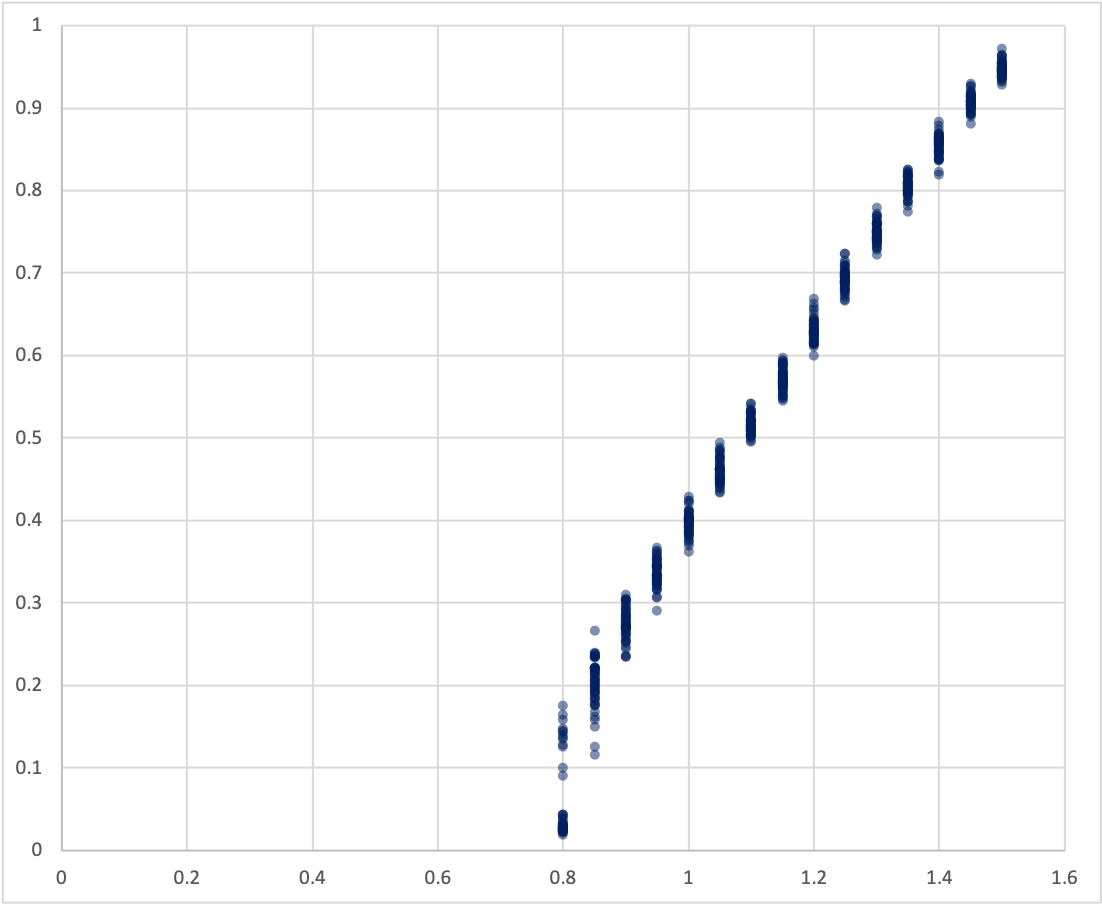}
	\caption{$k=3$}
	\label{fig:k3-concentration}
\end{subfigure}
	\caption{Concentration of the average error as a function of $\lambda=m/n$, for $k\in\{2,3\}$.}
\label{fig:concentration}
\end{figure}
Furthermore, we experimentally studied the empirical distribution of individual relative errors, which turns out to have an interesting multi-modal shape for intermediate values of~$\lambda$. Typical distributions for $k=2,3$ are illustrated in Figure~\ref{fig:distribution} where each point corresponds to an edge, and the edges are randomly ordered along the $x$-axis.  Each plot corresponds to an individual random graph. 

When $\lambda$ grows beyond the peelability threshold, a fraction of edges with small errors still remains but vanishes quickly: these include edges incident to at least one leaf (these have error $0$) and peelable edges (these have error $o(1)$, by arguments to be given in Section~\ref{sec:proofk3}). 
For intermediate values of $\lambda$, the distribution presents several modes: besides the main mode (largest concentration on plots of Figure~\ref{fig:distribution}), we observe a few other concentration values which are typically integers. While this phenomenon is still to be analysed, we explain it by the presence in the graphs of certain structural motifs that involve  disparities in node degrees. Note that the fraction of values concentrated around the main mode is dominant: for example, for $k=3,\lambda=3$ (Figure~\ref{fig:k3-l3}), about 90\% of values correspond to the main mode ($\approx 3.22$). Finally, when $\lambda$ becomes larger, these ``secondary modes'' disappear, and the distribution becomes concentrated around a single value. This is consistent with the tighter concentration observed earlier in Figure~\ref{fig:concentration}.


\begin{figure}[ht!]
	\centering
	\begin{subfigure}[b]{0.49\textwidth}
		\includegraphics[width=\textwidth]{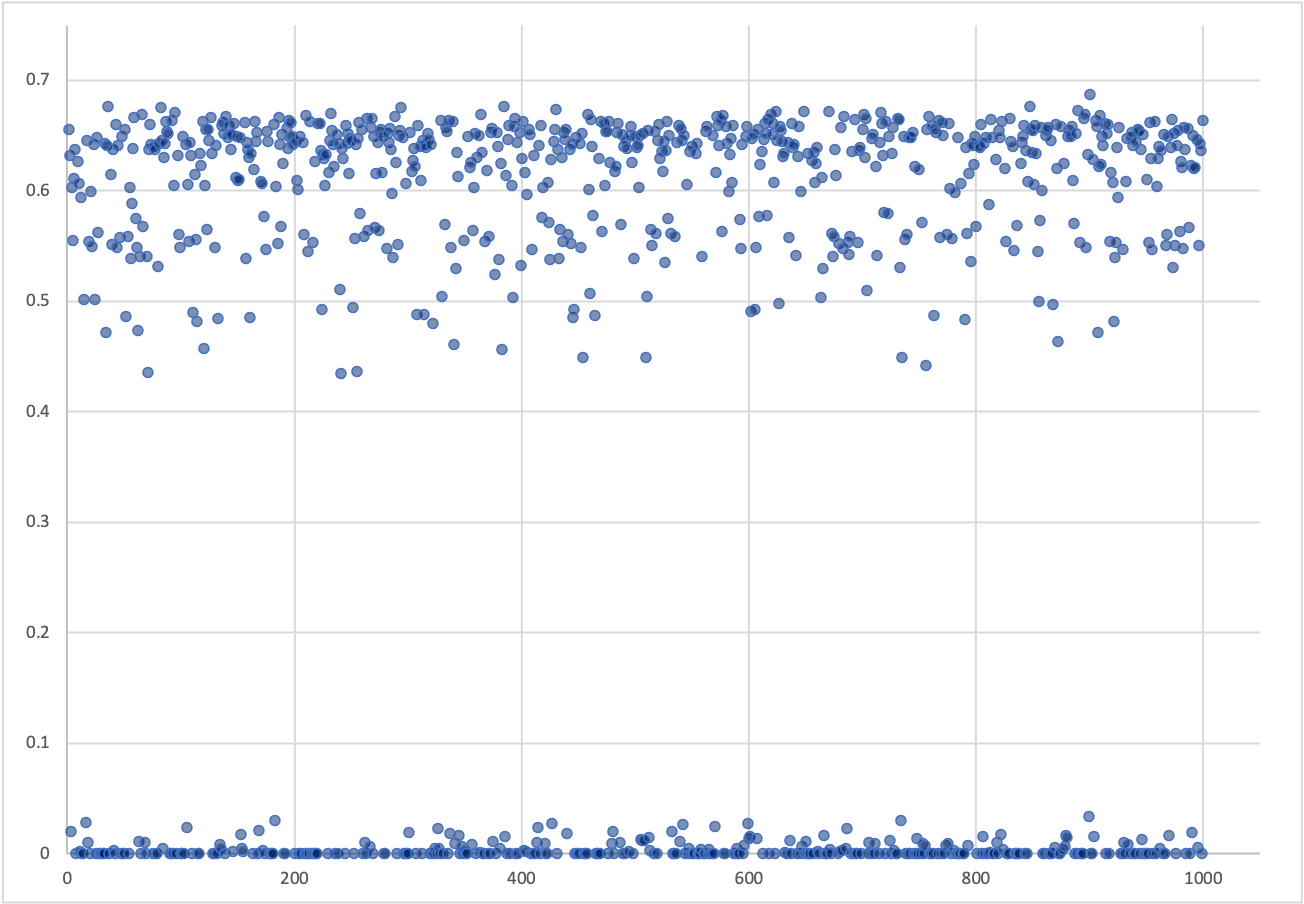}
		\caption{$k=2, \lambda=1$}
		\label{fig:k2-l1}
	\end{subfigure}
	\hfill
\begin{subfigure}[b]{0.49\textwidth}
	\centering
	\includegraphics[width=\textwidth]{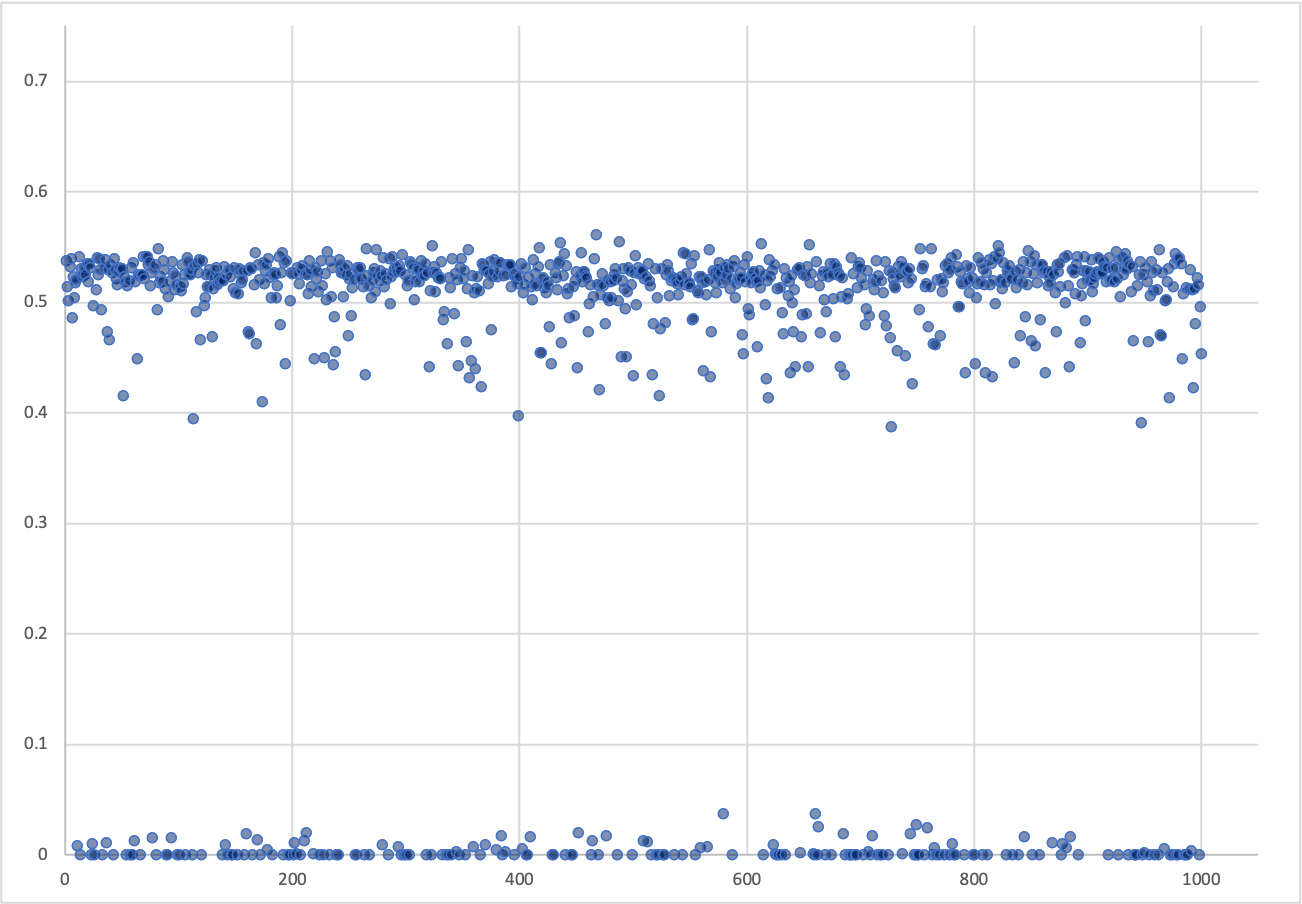}
	\caption{$k=3, \lambda=1$}
	\label{fig:k3-l1}
\end{subfigure}
\begin{subfigure}[b]{0.49\textwidth}
		\includegraphics[width=\textwidth]{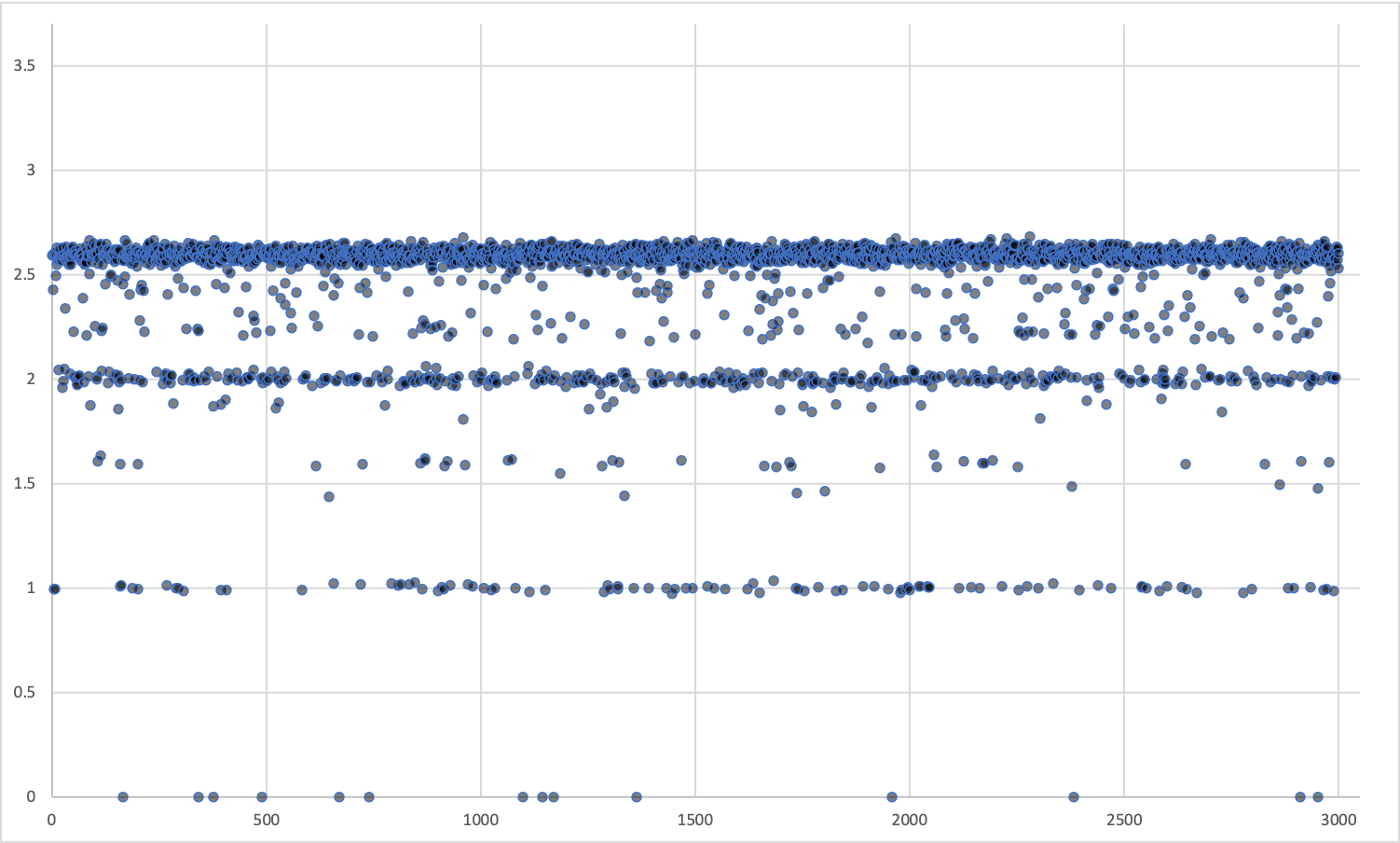}
		\caption{$k=2, \lambda=3$}
		\label{fig:k2-l3}
	\end{subfigure}
	\hfill
	\begin{subfigure}[b]{0.49\textwidth}
	\centering
	\includegraphics[width=\textwidth]{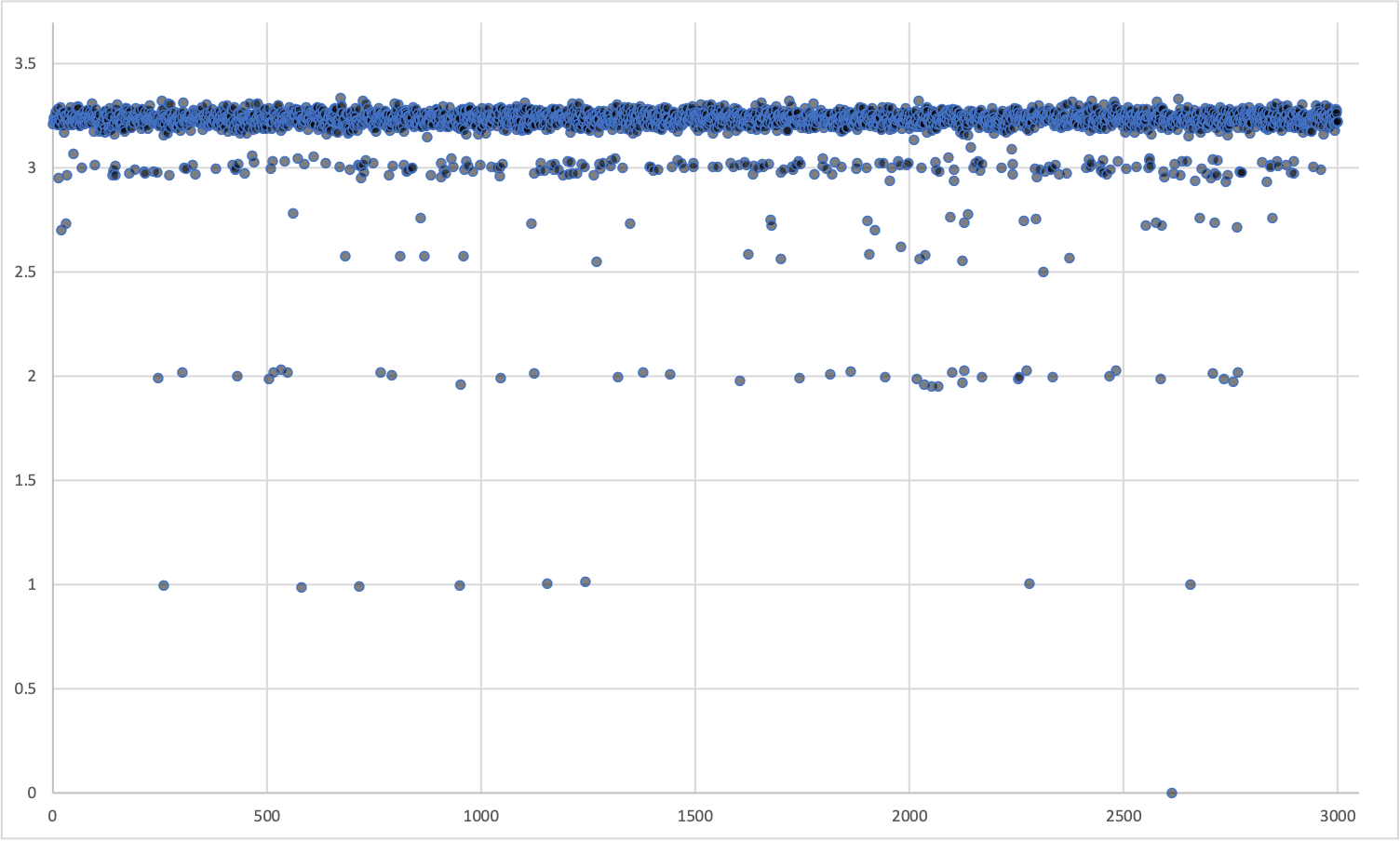}
	\caption{$k=3, \lambda=3$}
	\label{fig:k3-l3}
\end{subfigure}
\begin{subfigure}[b]{0.50\textwidth}
	\centering
	\includegraphics[height=129pt]{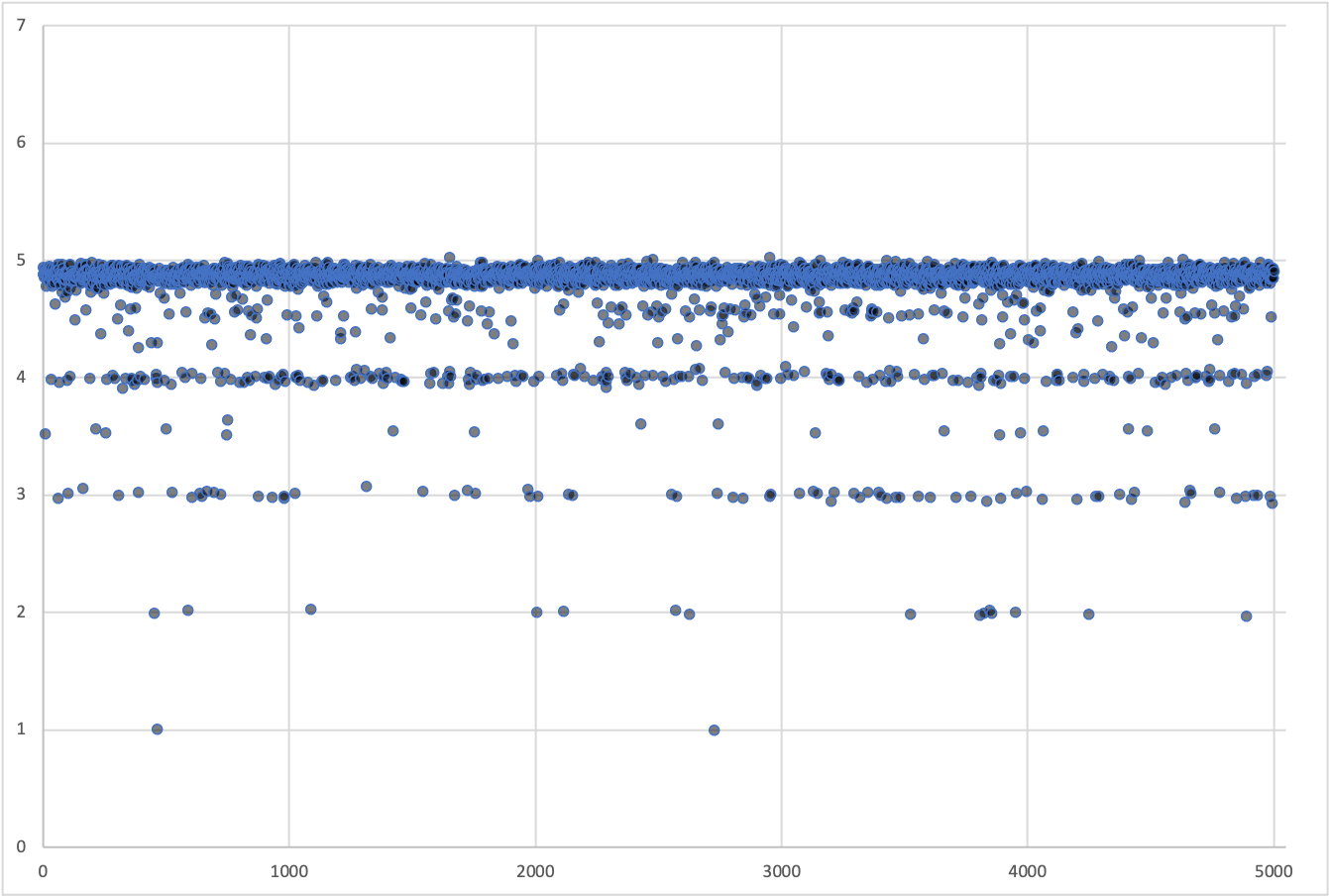}
	\caption{$k=2, \lambda=5$}
	\label{fig:k2-l5}
\end{subfigure}
\hfill
\begin{subfigure}[b]{0.48\textwidth}
	\centering
	\includegraphics[width=\textwidth]{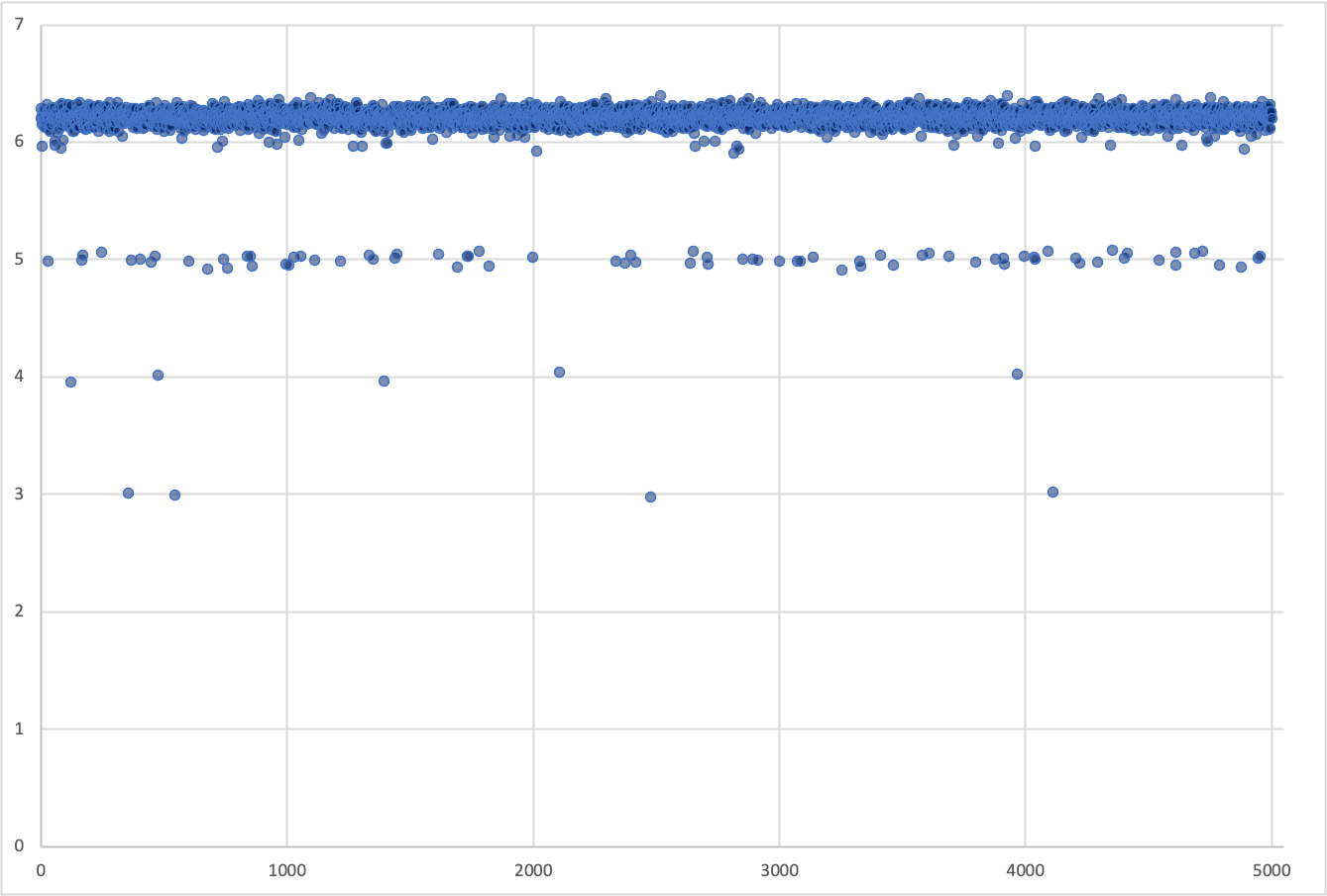}
	\caption{$k=3, \lambda=5$}
	\label{fig:k3-l5}
\end{subfigure}
	\caption{Distribution of relative errors of individual edges shown in a random order along $x$-axis.}
	\label{fig:distribution}
\end{figure}

Finally, we report on another experiment supporting the conjecture of a positive average error in the supercritical regime. We simulated the CU-process on sparse random non-peelable 3-hypergraphs (i.e. $k=3$), namely 2-regular 3-hypergraphs with $2n$ edges and $3n$ vertices ($n$ parameter). These are sparsest possible non-peelable 3-hypergraphs, with degree 2 of each vertex. We observed that the average error for such graphs is concentrated around a constant value of $\approx 0.217$. Since the core size is linear in the supercritical regime, this experiment provides an evidence of a positive error in the general case. While this remains to be proved in general, in Section~\ref{sec:non-peelable} we  provide a proof for certain families of regular hypergraphs. 


\newpage

\section{Proofs of main results}
\label{sec:proofs}
Theorem~\ref{theo:subcriti} relies on properties of random hypergraphs. Case $k=2$ corresponds to  Erd\H{o}s-R\'enyi random graphs $G_{n,\lambda n}$ \cite{erdos1960evolution} which have been extensively studied \cite{FriezeKaronski16}. In particular, it is well known 
when $\lambda<1/2$ and $n$ gets large, $G_{n,\lambda n}$ is, w.h.p., a union of small connected components most of which are constant-size trees. 
That is, a random edge in $G_{n,\lambda n}$ is, w.h.p., in a tree of size $O(1)$. Thus, 
the proof amounts to showing that, for a fixed tree $T$ and a vertex $v\in T$, we have $\cvN/N=1+o(1)$ w.h.p.. 
We prove this in Section~\ref{sec:error_tree} for both $N$-uniform and $N$-balanced models. 
The proof for $k\geq 3$, given in Section~\ref{sec:proofk3}, requires more ingredients. An additional difficulty is that, for $\lambda<\lambda_k$, a random edge $e$ in $H_{n,\lambda n}^k$ may be in the giant component (if $\lambda\in(\frac1{k(k-1)},\lambda_k)$). However, we rely on the fact that the peeling level of $e$ is $O(1)$ w.h.p., and prove that for a vertex $v$ of bounded level, we have $\cvN/N=1+o(1)$ w.h.p. as $N\to\infty$, where the $o(1)$ term does not depend on
the size of the giant component.

\subsection{CU-process on a fixed tree}\label{sec:error_tree}

\subsubsection{Analysis in the $N$-uniform model}\label{sec:tree_N_uniform}
Consider a graph $G=(V,E)$, with $m$ edges, on which the CU-process is run, in the $N$-uniform model. 
Recall that $c_v(t)$ (resp. $c_e(t)$) denotes the value of the counter for $v$  (resp. for $e$) after $t$ steps, and $o_e(t)$ is the number of occurrences $e$ in the first $t$ steps.   
The aim of this Section is to prove the following result.

\begin{lemma}\label{prop:cv}
Let $T=(V,E)$ be a tree, on which the CU-process is run, in the $N$-uniform model. Let $m=|E|$. Then, for every vertex $v$ of $T$,  
there exist absolute positive constants $a_v,b_v$ such that, for any $N\geq 1$ and $x>0$, we have
\[
\mathbb{P}\Big(\mathrm{max}_{t\in[0..Nm]}|c_v(t)-t/m| \geq x\sqrt{N}\Big)\leq a_v\exp(-b_vx^2).
\]
\end{lemma}
Lemma~\ref{prop:cv} implies that, in the $N$-uniform model,  
the final counter $\cvN$ of every vertex $v$ of $T$ is concentrated around $N$, with (sub-)Gaussian fluctuations  
of order $\sqrt{N}$. The same holds for the final counter $\ceN=\mathrm{min}(c_u^{(N)},c_v^{(N)})$ of every edge $e=(u,v)$ of $T$. On the other hand, 
the number $o_e(Nm)$ of times $e$ is chosen follows a binomial distribution $\mathrm{Bin}(Nm,1/m)$. Then, it is also concentrated around $N$, with Gaussian fluctuations
of order $\sqrt{N}$. This implies that $\ReN=O(1/\sqrt{N})=o(1)$ w.h.p. as $N$ gets large. 

We say that a family of events $\cE_M^N$ indexed by two parameters $N,M\geq 1$, is \emph{$(N,M)$-concentrated} 
if there are absolute constants $a,b>0$ such that, for every $N,M$, we have 
$\mathbb{P}(\cE_M^N)\leq a\exp(-bx^2)$, where $x=M/\sqrt{N}$. For $f(t)$ a (possibly random) quantity depending on $t\in[0..Nm]$, we use the notation $\ol{f}(t):=f(t)-t/m$. 
Thus, to prove Lemma~\ref{prop:cv}, we have  to show that 
for a fixed tree $T$ with $m$ edges on which the $N$-independent CU-process is run, and for $v$ a vertex of $T$, the event
$\{\mathrm{max}_{t\in[0..Nm]}|\ol{c_v}(t)|\geq M\}$ is $(N,M)$-concentrated. 

We proceed by induction on the peeling level $i$ of vertices. 
A vertex $v$ at level $0$ is a leaf. Let $e$ be its incident edge. 
It is easy to see that the counter of $v$ increases exactly at the steps when $e$ is drawn. Hence $c_v(t)=o_e(t)$ for $t\in[0..Nm]$.  
Doob's martingale maximal inequality combined with Hoeffding's inequality ensure 
 that, for every edge $e$ of $T$, we have 
 \[
 \mathbb{P}(\mathrm{max}_{t\in[0..Nm]}|\ol{o_e}(t)|\geq M)\leq 2\exp(-2x^2/m),\ \ \mathrm{where}\ x=M/\sqrt{N}.
 \] 
 Hence, for any leaf $v$ of $T$, the event $\{\mathrm{max}_{t\in[0..Nm]}|\ol{c_v}(t)|\geq M\}$ is $(N,M)$-concentrated.
 Moreover, for $v$ a vertex and $e$ an arbitrary edge incident to $v$, the fact that $c_v(t)\geq o_e(t)$  
 ensures that the event $\{\mathrm{min}_{t\in[0..Nm]}\ol{c_v}(t)) \leq -M\}$ is
 $(N,M)$-concentrated. It thus remains
  to show  that, for vertices of positive levels, the event $\{\mathrm{max}_{t\in[0..Nm]} \ol{c_v}(t) \geq M\}$ is $(N,M)$-concentrated.  
  
 The following statement will be useful to treat the inductive step. 
 \begin{lemma}\label{lem:dv}
Let $G$ be a graph on which a CU-process is run. Let $v$ be a vertex of $G$, with $v_1,\ldots,v_{h+1}$ its neighbours one of which (say $v_{h+1}$) is distinguished, and with $e$ denoting the edge $\{v,v_{h+1}\}$.  
For $t\in[0..Nm]$, let  $d_v(t)=\mathrm{max}(c_{v_1}(t),\ldots,c_{v_{h}}(t))$. Consider the event $\cE_M^N$
that there exists $t\in[0..Nm]$ such that $c_v(t)\geq d_v(t)+M$, and the event $\cF_M^N$ that there exists $t\in [0..Nm]$ such that $|\ol{o_e}(t)|\geq M/4$ or $|\ol{d_v}(t)|\geq M/4$. Then $\cE_M^N$ implies $\cF_M^N$.   
\end{lemma}
\begin{proof}
If $\cE_M^N$ holds, let $t_0\geq 0$ be such that $c_v(t_0)\geq d_v(t_0)+M$. Let $t'=\mathrm{max}\{t\leq t_0\ |\ c_v(t)\leq d_v(t)\}$. The crucial point is that, in the interval $[t'..t_0]$, any step where $c_v$ increases occurs when $e$ is chosen (indeed, when $c_v(t)>\mathrm{max}(c_{v_1}(t),\ldots,c_{v_{h}}(t))$, choosing an edge $\{v,v_i\}$ with $i\in[1..h]$ yields no increase of $c_v$). Hence,  
$c_v(t_0)-c_v(t')\leq o_e(t_0)-o_e(t')$. Since $c_v(t')=d_v(t')$ and $c_v(t_0)\geq d_v(t_0)+M$ we conclude that $o_e(t_0)-o_e(t')\geq d_v(t_0)-d_v(t')+M$. Thus we also have 
$\ol{o_e}(t_0)-\ol{o_e}(t')\geq \ol{d_v}(t_0)-\ol{d_v}(t')+M$. Hence $\mathrm{max}(|\ol{o_e}(t_0)|, |\ol{o_e}(t')|,|\ol{d_v}(t_0)|,|\ol{d_v}(t')|)\geq M/4$, so that $\cF_M^N$ holds.  
\end{proof}

Let $i\geq 1$, and assume Lemma~\ref{prop:cv} holds for all vertices of level smaller than $i$. Let $v$ be a vertex of level $i$, for which we want to prove that the event $\{\mathrm{max}_{t\in[0..Nm]} \ol{c_v}(t) \geq M\}$ is $(N,M)$-concentrated. 
All neighbours $v_1,\ldots,v_{h+1}$ of $v$ have level at most $i-1$, except for one, say $v_{h+1}$, 
with its level in $\{i-1,i,i+1\}$ (respectively corresponding to $T_i$ having one vertex, two vertices, or at least three vertices).  

Let $e$ be the edge between $v$ and $v_{h+1}$. Let $\cE_M^N$ be the event that $\ol{c_v}(t)\geq M$ for some $t\in [0..Nm]$. 
If this holds, then one of the two events $\{\ol{c_v}(t)-\ol{d_v}(t)\geq M/2\}$ or $\{\ol{d_v}(t)\geq M/2\}$ holds. 
By Lemma~\ref{lem:dv}, the first event implies that $|\ol{d_v}(t)|\geq M/8$ or $|\ol{o_e}(t)|\geq M/8$ for some $t\in[0..Nm]$. Hence, the event 
$\{\max_{t\in[0..Nm]}\ol{c_v}(t)\geq M\}$ 
 implies that either the event  
 $\{\max_{t\in[0..Nm]}|\ol{d_v}(t)|\geq M/8\}$ (which is also the union of the events 
 $\{\max_{t\in[0..Nm]}|\ol{c_{v_j}}(t)|\geq M/8\}$ for $j\in[1..h]$) or the event $\{\max_{t\in[0..Nm]}|\ol{o_e}(t)|\geq M/8\}$ holds.
 Since these events are $(N,M)$-concentrated, we conclude (using the union bound) that the event  $\{\max_{t\in[0..Nm]}\ol{c_v}(t)\geq M\}$ is also $(N,M)$-concentrated, which concludes the proof of Lemma~\ref{prop:cv}.
 
\subsubsection{Analysis in the $N$-balanced model}\label{sec:tree_N_balanced}

\begin{lemma}\label{prop:tree_N_uniform}
Let $T$ be a tree, on which the CU-process is performed, in the $N$-balanced model. Then, for every edge $e$ of $T$, we have $\ReN=o(1)$ w.h.p. as $N\to\infty$. 
\end{lemma}
\begin{proof}
Note that the $N$-balanced model is just the $N$-uniform model conditioned on all edges occurring exactly $N$ times, which happens with probability $\Theta(n^{-m/2})$ if $T$ has $m$ edges. 
Let $u$ be an extremity of $e$. 
By Lemma~\ref{prop:cv}, there exists a positive constant $b$ such that, in the $N$-uniform model,
\[
\mathbb{P}(c_u(Nm)\geq N+N^{2/3})=O(e^{-b\, n^{1/3}}).
\] 
Hence, in the $N$-balanced model, we have 
\[
\mathbb{P}(c_u^{(N)}\geq N+N^{2/3})=O(n^{m/2}e^{-b\, n^{1/3}})=o(1),
\] 
and thus
\[
\mathbb{P}(\ReN\geq N^{-1/3})=\mathbb{P}(\ceN\geq N+N^{2/3})\leq \mathbb{P}(c_u^{(N)}\geq N+N^{2/3})=o(1),
\]
which ensures that $\ReN=o(1)$ w.h.p. as $N\to\infty$. 
\end{proof}

\subsection{Proof of Theorem~\ref{theo:subcriti} for $k=2$}\label{sec:subcritical2}
We use the well-known property~\cite{erdos1960evolution} that, for fixed $\lambda\in(0,1/2)$,  a random edge $e$ in $G_{n,\lambda n}$ is w.h.p. in a tree of size $O(1)$. 
Precisely, if we let $\cE^m_{n,\lambda}$ be the event that $e$ belongs to a tree of size at most $m$, then we have the property that, for every $\epsilon>0$,
there exist $m$ and $n_0$ such that $\mathbb{P}(\cE^m_{n,\lambda})\geq 1-\epsilon$ for $n\geq n_0$. 
 By Lemma~\ref{prop:tree_N_uniform}, there exists $N_0$ (depending on $m$) such that, in the $N$-balanced or $N$-uniform model,  
for every tree $T$ with at most $m$ edges, and every edge $e\in T$, 
we have $\mathbb{P}(\ReN\geq\epsilon)\leq \epsilon$ for $N\geq N_0$. 
Hence, for $n\geq n_0$ and $N\geq N_0$,  
if we perform the CU-process on $G_{n,\lambda n}$ in the $N$-balanced model (note that the $N$-balanced model holds separately on every connected component),  
and draw a random edge $e$, 
we have 
\[
\mathbb{P}(\ReN\geq\epsilon)\leq\mathbb{P}(\ReN\geq\epsilon\ |\ \cE^m_{n,\lambda})+\mathbb{P}(\neg\cE^m_{n,\lambda})\leq 2\epsilon.
\]
This means that $\ReN=o(1)$ w.h.p., when $n$ and $N$ grow.  In the $N$-uniform model, the same argument holds, using the fact that the total number of times 
 edges in $T$ are drawn is concentrated around $Nm$.

\subsection{Proof of Theorem~\ref{theo:subcriti} for $k\geq 3$}\label{sec:proofk3}

The proof partly follows the same lines as for $k=2$, but requires additional arguments, in particular a suitably extended notion of peelability. 
A \emph{marked hypergraph} is a hypergraph $H=(V,E)$ where some of the vertices are marked. The subset of marked vertices is denoted $V_{\infty}$. When performing a CU-process on a marked hypergraph, the counters of unmarked vertices are (as usual) initially $0$, while the counters at marked vertices are initially (and remain) $+\infty$. 
When peeling, the marked vertices are not allowed to be peeled (even when they are incident to a unique edge). We define $H_0=H$, and iteratively for $i\geq 0$, we define $V_i$  as the set of non-marked vertices that are leaves or isolated vertices in $H_i$, 
$E_i$ as the set of edges of $H_i$ incident to vertices in $V_i$, and $H_{i+1}$ to be the hypergraph obtained from $H_i$ by deleting 
the vertices in $V_i$ and the edges in $E_i$. A vertex in $V_i$ is said to have \emph{level $i$}. 
Then $H$ is called \emph{peelable} if every unmarked vertex is peeled at some step. 
 
 Following the exact same lines as in the proof of Lemma~\ref{prop:cv} (induction on the vertex-levels, and a straightforward adaptation of Lemma~\ref{lem:dv} to hypergraphs) and Lemma~\ref{prop:tree_N_uniform}, we obtain:

\begin{lemma}\label{prop:cvH}
Let $H=(V,E)$ be a peelable connected marked hypergraph, on which the CU-process is run, in the $N$-uniform model. Let $m=|E|$, and for $t\in[0..Nm]$, let $c_v(t)$ be the 
value of the counter of $v$ after $t$ steps.  
Then, for every unmarked vertex $v$ of $H$,    
there exist absolute positive constants $a_v,b_v$ such that, for any $N\geq 1$ and $x>0$, we have
\[
\mathbb{P}\Big(\mathrm{max}_{t\in[0..Nm]}|c_v(t)-t/m| \geq x\sqrt{N}\Big)\leq a_v\exp(-b_vx^2).
\]
In the $N$-uniform or $N$-balanced model, for every unmarked vertex $v$ of $H$, we have $c_v^{(N)}/N=1+o(1)$ w.h.p. as $N\to\infty$.
\end{lemma}  

For $H$ a hypergraph, and for $v$ a vertex of finite level in $H$, a vertex $u\in H$ is called a \emph{descendant} of $v$ if there is a sequence $v_0=u,v_1,\ldots,v_r=v$ of vertices such that, for each $j\in[1..r-1]$, $v_{j}$ and $v_{j+1}$ are incident to a same edge, and the level of $v_{j+1}$ is larger than the level of $v_j$. 

\begin{lemma}\label{lem:Hv}
For a hypergraph $H$, and for a vertex $v$ of $H$ of finite level, let $D_v$ be the set formed by $v$ and it descendants, and let $E_v$ be the set of edges
that are incident to at least one vertex from $D_v$. Let $H_v$ be the marked hypergraph formed by the edges in $E_v$ and their incident vertices, where the marked vertices
are those not in $D_v$.  Then $H_v$ is peelable.
\end{lemma}
\begin{proof}
Let $u$ be a vertex in $D_v$, of level $i$. If $i=0$, then $u$ is a leaf in $H$. It is also a leaf in $H_v$, and is immediately peeled. 
Otherwise, except for possibly one incident edge, $u$ has at least one neighbour $v_e$ of level smaller than $i$  in every incident edge $e$ (note that $v_e$ has to be in $D_v$). 
By induction, $v_e$ and $e$ are peeled in the peeling  process of $H_v$. Hence, $u$ is peeled as well during the peeling process on $H_v$ (by induction as well, the level of a vertex in $D_v$ is actually the same in $H$ and in $H_v$). 
\end{proof}

\begin{remark}\label{rk:shell_ball}
In a hypergraph $H=(V,E)$, the \emph{distance} between two vertices $u$ and $v$ is the minimum number of edges to traverse in order to reach $v$ from $u$.
For $r\geq 1$ and $v\in V$, the ball $B_H(v,r)$ is the set of vertices at distance at most $r$ from $v$. Clearly, if $v$ is at level $i$, then every vertex in $D_v$ is in 
$B_H(v,i)$, and every vertex of $H_v$ is in $B(v,i+1)$. \dotfill
\end{remark}

\begin{lemma}[Lemma~3 in \cite{molloy2005cores}]\label{lem:I}
Let $k\geq 3$, and let $\lambda\in(0,\lambda_k)$. For every $\epsilon>0$, there exist $I(\epsilon)$ and $n_0(\epsilon)$ such that, 
for $v$ a random vertex in $H^k_{n,\lambda n}$, the probability that $v$ has level at most $I(\epsilon)$ is at least $1-\epsilon$, for $n\geq n_0(\epsilon)$.
\end{lemma}

\begin{lemma}\label{lem:monoton}
Let $H=(V,E)$ be a hypergraph, and let $E'$ be a subset of $E$. 
Let $V'$ (resp. $W$) be the subset of vertices incident to at least one edge from $E'$ (resp. incident only to edges in $E'$). Let $H'=(V',E')$ be the marked
hypergraph where the unmarked vertices are those in $W$.  
Consider a CU-process on $H$, with $S$ the sequence of items (each item an edge of $H$),  and for $v\in V$ let $c_v$ be the final value of the counter of $v$. Let $S'$ 
be the subsequence of $S$ composed by the items in $E'$. Consider the CU process on $H'$ where the sequence of    items is $S'$.  
For $v\in W$, let $c_v'$ be the final value of the counter of $v$. Then $c_v'\geq c_v$. 
\end{lemma}
\begin{proof}
For an assignment $f:V\to \mathbb{N}\cup\{+\infty\}$ of initial values to the vertex counters, let $c_f:V\to\mathbb{N}\cup\{+\infty\}$ be the function giving the final vertex counters after performing the CU-process on an input sequence $S$. 
It is easy to check (by induction on the length of $S$) that the CU-process is monotonous: if $f(v)\leq g(v)$ for all $v\in V$, then $c_f(v)\leq c_g(v)$ for all $v\in V$. 
Now, if we define $f$ as the function assigning initial value $0$ to vertices in $W$,
and initial value $+\infty$ to the other vertices, then  $c_v'=c_f(v)$ for all $v\in W$. On the other hand, if $f_0$ denotes the function assigning initial value $0$ to all vertices, then
$c_v=c_{f_0}(v)$ for all $v\in V$. Hence, the monotonicity property ensures that $c_v'\geq c_v$ for all $v\in W$.
\end{proof}

We can now conclude the proof of Theorem~\ref{theo:subcriti} for $k\geq 3$. Let $\epsilon>0$. Let $I(\epsilon)$ and $n_0(\epsilon)$ be as given in Lemma~\ref{lem:I}.   
Let $s(\epsilon)$ be such that, for a random vertex $v\in H_{n,\lambda n}^k$, we have $\mathbb{P}\big(|B(v,I(\epsilon)+1)|\leq s(\epsilon)\big)\geq 1-\epsilon$ for all $n$.  
The existence of $s(\epsilon)$ easily follows from the known property that $H_{n,\lambda n}$ converges to a Galton-Watson branching
 process for the local topology, see e.g.~\cite[Prop 2.6]{dembo2010gibbs}.      
Let $v$ be a random vertex in $H_{n,\lambda n}$. Let $\cE_{n,\epsilon}$ 
be the event that the level of $v$ is at most $I(\epsilon)$, and the number of vertices in $H_v$ is at most $s(\epsilon)$.
Then, using Remark~\ref{rk:shell_ball} and the union bound, we have $\mathbb{P}(\cE_{n,\epsilon})\geq 1-2\epsilon$ for $n\geq n_0(\epsilon)$. 
From Lemma~\ref{prop:cvH}, there exists $N_0$ such that, for every marked $k$-uniform hypergraph $H$ with at most $s(\epsilon)$ vertices, and for every unmarked vertex $v$ 
of $H$, we have $\displaystyle\mathbb{P}\big(\cvN/N\ -1\geq \epsilon\big)\leq \epsilon$ for $N\geq N_0$. 
Let $e$ be a random edge in $H^k_{n,\lambda n}$, and let $v$ be a random vertex of $e$. Note that $v$ is distributed as a random vertex in $H_{n,\lambda n}^k$. 
With Lemma~\ref{lem:Hv} and Lemma~\ref{lem:monoton}, this implies that, 
conditioned on $\cE_{n,\epsilon}$, for  $N\geq N_0$ we have $\displaystyle\mathbb{P}\big(\cvN/N\ -1\geq \epsilon\big)\leq \epsilon$ in the $N$-balanced model, and also 
$\mathbb{P}(\ReN\geq \epsilon)\leq \epsilon$ (since $\ceN \leq \cvN$). Since $\mathbb{P}(\neg \cE_{n,\epsilon})\leq 2\epsilon$, we conclude that, for $n\geq n_0(\epsilon)$ and $N\geq N_0$, we have
$\mathbb{P}(\ReN\geq \epsilon)\leq 3\epsilon$. 
Hence, in the $N$-balanced model, $\ReN=o(1)$ w.h.p., when $n$ and $N$ grow. 
In the $N$-uniform model, the same argument holds again, using the fact that the number of times an edge in $H_v$ is drawn is concentrated around $Nm$, with $m$ the 
number of edges of $H_v$.

\subsection{Proof of Theorem~\ref{theo:supercrit}}
The \emph{excess} of a graph $G$ is $\mathrm{exc}(G)=|E|-|V|$. 
\begin{lemma}\label{lem:excess}
Let $G=(V,E)$ be a graph. Then, for the $N$-balanced model, we have 
 $\sum_{e\in E}\ReN\geq \frac1{2}\,\mathrm{exc}(G)$.
\end{lemma}
\begin{proof}
During the CU process, each time an edge is drawn, the counter of at least one of its extremities is increased by $1$. Hence $\sum_{v\in V}\cvN\geq N |E|$. 
Hence, with the notation $R_v^{(N)}:=\cvN/N-1$, we have $\sum_{v\in V}R_v^{(N)}\geq \mathrm{exc}(G)$. Now, by Lemma~\ref{lem:observation}, for each $v\in V$, there exists an edge $e_v$ incident to $v$ such that $c_{e_v}^{(N)}=\cvN$ (if several incident edges have this property, an arbitrary one is chosen). Hence,  $\sum_{v\in V}R_{e_v}^{(N)}\geq \mathrm{exc}(G)$. Note that, in this sum, every edge occurs at most twice (since it has two extremities), thus $\sum_{e\in E}\ReN\geq \frac1{2}\mathrm{exc}(G)$. 
\end{proof}

For $\lambda>1/2$, it is known~\cite[Theorem~6]{pittel2005counting} that there is an explicit constant $\tilde{f}(\lambda)>0$ such that the excess of the giant component $G'=(V',E')$ of 
$G_{n,\lambda n}$  is concentrated around $\tilde{f}(\lambda)n$, with fluctuations of order $\sqrt{n}$. Thus, $\mathrm{exc}(G')\geq \frac1{2}\tilde{f}(\lambda)n$ 
w.h.p. as $n\to\infty$.  
Hence, by Lemma~\ref{lem:excess}, w.h.p. as $n\to\infty$  (and for any $N\geq 1$), we have
\[
\mathrm{err}_N(G_{n,\lambda n})\geq \frac1{\lambda n}\sum_{e\in E'}\ReN\geq \frac1{2\lambda n}\,\mathrm{exc}(G') \geq \frac1{4\lambda} \tilde{f}(\lambda)=:f(\lambda).
\]

\section{Analysis for some non-peelable hypergraphs}
\label{sec:non-peelable}


Analysing the asymptotic behaviour of the relative error of the CU-process on arbitrary hypergraphs seems to be a challenging task,
 even if we restrict ourselves to $N$-uniform and $N$-balanced models, as we do in this paper.  Based on simulations, we expect that, for a fixed connected 
$k$-hypergraph $H=(V,E)$,  and for  $v\in V$, we have $\cvN/N=C_v+o(1)$ w.h.p. as $N\to\infty$, for an explicit constant $C_v\in[1,\mathrm{deg}(v)]$. Since the number of increments at each step lies in $[1..k]$, constants $C_v$ must verify $1\leq \frac1{|E|}\sum_{v\in V} C_v\leq k$, where $\frac1{|E|}\sum_{v\in V} C_v$ can be seen as the average number of increments at a step. 
 If $H$ is peelable,
then Lemma~\ref{prop:cvH} implies that
this concentration holds, with $C_v=1$. We expect that, if no vertex of $H$ is peelable, and if $H$ is ``sufficiently homogeneous'', then the constants $C_v$ should be all equal to the same constant $C>1$, and thus the relative error $\ReN$ of every edge is concentrated around $C-1>0$ w.h.p. as $N\to\infty$.  This, in particular, is supported by an experiment reported at the end of Section~\ref{sec:simulation}.

In this Section, we show that this is the case for a 
  family of regular hypergraphs which are very sparse ($O(\sqrt{|V|})$ edges) but have a high order (an edge contains $O(\sqrt{|V|})$ vertices). 
%
The \emph{dual} of a hypergraph $H$ is the hypergraph $H'$ where the roles of vertices and edges are interchanged:
the vertices of $H'$ are the edges of $H$, and the edges of $H'$ are the vertices of $H$
so that an edge of $H'$ corresponding to a vertex $v$ of $H$ contains those vertices that correspond to edges incident to $v$ in $H$. 
We consider here the hypergraph $K_{n}'$ dual to the complete graph $K_{n}$. 


\begin{lemma}\label{lem:complete}
Consider any CU-process on $K_{n}'=(V,E)$. At each step $t$, let $\mathrm{minedges}(t)$ be the set of edges of $E$ whose counter $c_e(t)$ is minimal over all edges. 
We have $|\mathrm{minedges}(t)|\geq 2$. 
Let $e$ be the edge selected at time $t$. Then, at time $t$, $e$ is the only edge whose counter increases (by $1$), except when $e\in\mathrm{minedges}(t)$
and $|\mathrm{minedges}(t)|=2$, in which case the counter of the other edge in $\mathrm{minedges}(t)$ also increases (by $1$). 
\end{lemma}
\begin{proof}
Every vertex $v\in V$ has degree~$2$. Considering a vertex $v$ of minimal counter, its two incident edges thus
have to be in $\mathrm{minedges}(t)$, therefore  $|\mathrm{minedges}(t)|\geq 2$.  

Let $e$ be the chosen edge at time~$t$. Clearly, $c_e$ increases (by $1$). 
Moreover, for every edge $e'$ such that $c_{e'}(t)\neq c_e(t)$, $c_{e'}$ does not increase.    
Let now $e'$ be such that $c_{e'}(t)=c_e(t)$, but there exists an edge $e''\notin\{e,e'\}$ with $c_{e''}(t)\leq c_{e'}(t)$. Let $v=e'\cap e''$. By Lemma~\ref{lem:observation}, we must have $c_v(t)=c_{e'}(t)$.  
Since $v$ is not incident to $e$, its counter does not increase, and neither does the counter of~$e'$.  
If there does not exist $e''\notin\{e,e'\}$ with $c_{e''}(t)\leq c_{e'}(t)$, then we must be in the situation where  $e\in\mathrm{minedges}(t)$ and $|\mathrm{minedges}(t)|=2$.
Let $v=e\cap e'$. Then, by Lemma~\ref{lem:observation}, we have
$c_v(t)=c_e(t)=c_{e'}(t)$, and any vertex $v'\neq v$ must satisfy $c_{v'}(t)>c_v(t)$ (since $v'$ has at least one incident edge different from $e,e'$). 
Thus, the counter of $v$ increases, and the counter of $e'$ also increases. 
\end{proof}

\begin{theorem}\label{theo:dual_complete}
For any fixed $n\geq 2$, in the $N$-uniform model (resp. in the $N$-balanced model), {the counter of each vertex $v\in K_n'$ satisfies $\cvN/N=n/(n-1)+o(1)$ w.h.p. as $N\to\infty$.}  
Hence, the relative error $\ReN$ of each edge $e$ in $K_{n}'$ satisfies $\ReN=1/(n-1)+o(1)$ w.h.p. as $N\to\infty$. 
\end{theorem}
\begin{proof}
Consider a CU-process performed on $K_{n}'=(V,E)$. At each step $t$, let $p(t)=\mathrm{min}(c_v(t),\ v\in V)=\mathrm{min}(c_e(t),\ e\in E)$, 
and let $q(t)=\mathrm{max}(c_e(t),\ e\in E)$.  
It easily follows from Lemma~\ref{lem:complete} that $\sum_{e\in E}c_e(t)=t+p(t)$. Therefore, $(n-1)p(t)\leq t\leq (n-1)q(t)$, so that $\frac{t}{n-1}\in[p(t),q(t)]$.  
We now consider the CU process on $K_{n}'$ in the $N$-uniform model. 
Let $e_1,e_2$ be two edges in $E$, and let $M\geq 1$. For $t\in[0..Nn]$, the event that $c_{e_1}(t)\geq c_{e_2}(t)+M$  
implies the event that $o_{e_1}(t)-o_{e_1}(t')\geq o_{e_2}(t)-o_{e_1}(t')+M$ for some $t'\in[0..t]$ (indeed, letting $t'$ be the last time before $t$ where $c_{e_1}=c_{e_2}$,
Lemma~\ref{lem:complete} ensures that all times in $[t'..t-1]$ where $c_{e_1}$ increases are due to chosing $e_1$), which 
is the same as the event that $\ol{o_{e_1}}(t)-\ol{o_{e_1}}(t')\geq \ol{o_{e_2}}(t)-\ol{o_{e_1}}(t')+M$ for some $t'\in[0..t]$, with the notation of Section~\ref{sec:tree_N_uniform}.  
Therefore, the event that $c_{e_2}(t)-c_{e_1}(t)\geq M$ for some $t\in[0..Nn]$ implies the event that $\mathrm{max}(|\ol{o_{e_1}}(t)|,|\ol{o_{e_2}}(t)|)\geq M/4$ for some $t\in[0..Nn]$.    
With the terminology of Section~\ref{sec:tree_N_uniform}, for every $e\in E$, the event $\mathrm{max}_{t\in [0,Nn]}|\ol{o_e}(t)|\geq M$ is $(N,M)$-concentrated. Hence, the event
$\{\mathrm{max}_{t\in [0,Nn]}(c_{e_2}(t)-c_{e_1}(t))\geq M\}$ is $(N,M)$-concentrated. This implies that $\mathbb{P}(c_{e_2}(Nn)-c_{e_1}(Nn)\geq N^{2/3})=O(e^{-bN^{1/3}})$ for some $b>0$, and (by the union bound) that $\mathbb{P}(q(Nn)-p(Nn)\geq N^{2/3})=O(e^{-\tilde{b}N^{1/3}})$ for some $\tilde{b}>0$. Note that $\frac{Nn}{n-1}\in[p(Nn),q(Nn)]$, and 
for every edge $e\in E$ we have $c_e(Nn)\in  [p(Nn),q(Nn)]$. Therefore, $\mathbb{P}(|c_e(Nn)-\frac{Nn}{n-1}|\geq N^{2/3})=O(e^{-\tilde{b}N^{1/3}})$. Hence, in the $N$-uniform 
model we have $\ReN-\frac1{n-1}=O(N^{-1/3})=o(1)$ w.h.p. as $N\to\infty$.   
In the $N$-balanced model, by the same argument as in Lemma~\ref{sec:tree_N_balanced},  we have 
$\mathbb{P}(|\ceN-\frac{Nn}{n-1}|\geq N^{2/3})=O(N^{n/2}e^{-\tilde{b}N^{1/3}})$; hence, $\ReN-\frac1{n-1}=O(N^{-1/3})=o(1)$ w.h.p. as $N\to\infty$.   
\end{proof}

Note that with regular Count-Min, the relative error of every edge is $1+o(1)$ (exactly $1$ in the $N$-balanced model), 
    since all vertices have degree 2.   
The statement and proof of Theorem~\ref{theo:dual_complete} easily extend to dual complete hypergraphs as follows. Let $K_{n,r}$ be the complete $r$-uniform hypergraph on $n$ vertices (for $n\geq r\geq 2$), and consider a CU-process on $K_{n,r}'$ in the $N$-uniform (resp. $N$-balanced) model. 
Then { the counter of each vertex $v\in K_{n,r}'$ satisfies $\cvN/N=n/(n-r+1)+o(1)$ w.h.p as $N\to\infty$}, hence the relative error $\ReN$ of any edge $e\in K_{n,r}'$ satisfies $\ReN=\frac{r-1}{n-r+1}+o(1)$ w.h.p. as $N\to\infty$.  
For regular Count-Min, the relative error is $r-1+o(1)$ (exactly $r-1$ in the $N$-balanced model), since every vertex in $K_{n,r}'$ has degree $r$. 













\section{Non-uniform distributions}
\label{sec:zipf}

An interesting  and natural question is whether the phase transition phenomenon holds for non-uniform distributions as well. This question is of practical importance, as in many practical situations keys are not distributed uniformly. In particular, Zipfian distributions often occur in various applications and are a common test case for Count-Min sketches \cite{CormodeMuthuICDM05,DBLP:conf/teletraffic/BianchiDLS12,DBLP:conf/iccnc/EinzigerF15,10.1145/3487552.3487856,DBLP:journals/corr/abs-2203-14549}. 

In Zipf's distributions, key probabilities in descending order are proportional to $1/i^\beta$, where $i$ is the rank of the key and $\beta\geq 0$ is the \textit{skewness} parameter. Note that for $\beta = 0$, Zipf's distribution reduces to the uniform one. It is therefore a natural question whether the phase transition occurs for Zipf's distributions with $\beta > 0$. 

One may hypothesize that the answer to the question should be positive, as under Zipf's distribution, frequent keys tend to have no error, as it has been observed in earlier papers \cite{DBLP:conf/teletraffic/BianchiDLS12,DBLP:journals/corr/abs-2203-14549,benmazziane:hal-03613957}. On the other hand, keys of the tail of the distribution have fairly similar frequencies, and therefore might show the same behavior as for the uniform case. 

However, this hypothesis does not hold. Figure~\ref{fig:zipf} shows the behavior of the average error for Zipf's distributions with $\beta \in \{0.2,0.5,0.7,0.9\}$ vs. the uniform distribution ($\beta=0$). The average error is defined here as the average error of all keys weighted by their frequencies\footnote{This definition is natural for non-uniform distributions, as the error for a frequent key should have a larger contribution. Note that it is consistent with the definition of Section~\ref{sec:main_results} in the $N$-balanced case, and in the $N$-uniform case it presents a negligible difference when $N$ gets large.}, i.e. $\mathrm{err}_N(H)=\frac{1}{mN}\sum_{e\in E}\oeN\frac{\ceN-\oeN}{\oeN}=\frac{1}{mN}\sum_{e\in E}(\ceN-\oeN)$. In other words, $\mathrm{err}_N(H)$ is the expected error of a randomly drawn key from the entire input stream of length $mN$ (taking into account multiplicities). 

\begin{figure}
	\begin{center}
		\includegraphics[width=0.7\textwidth]{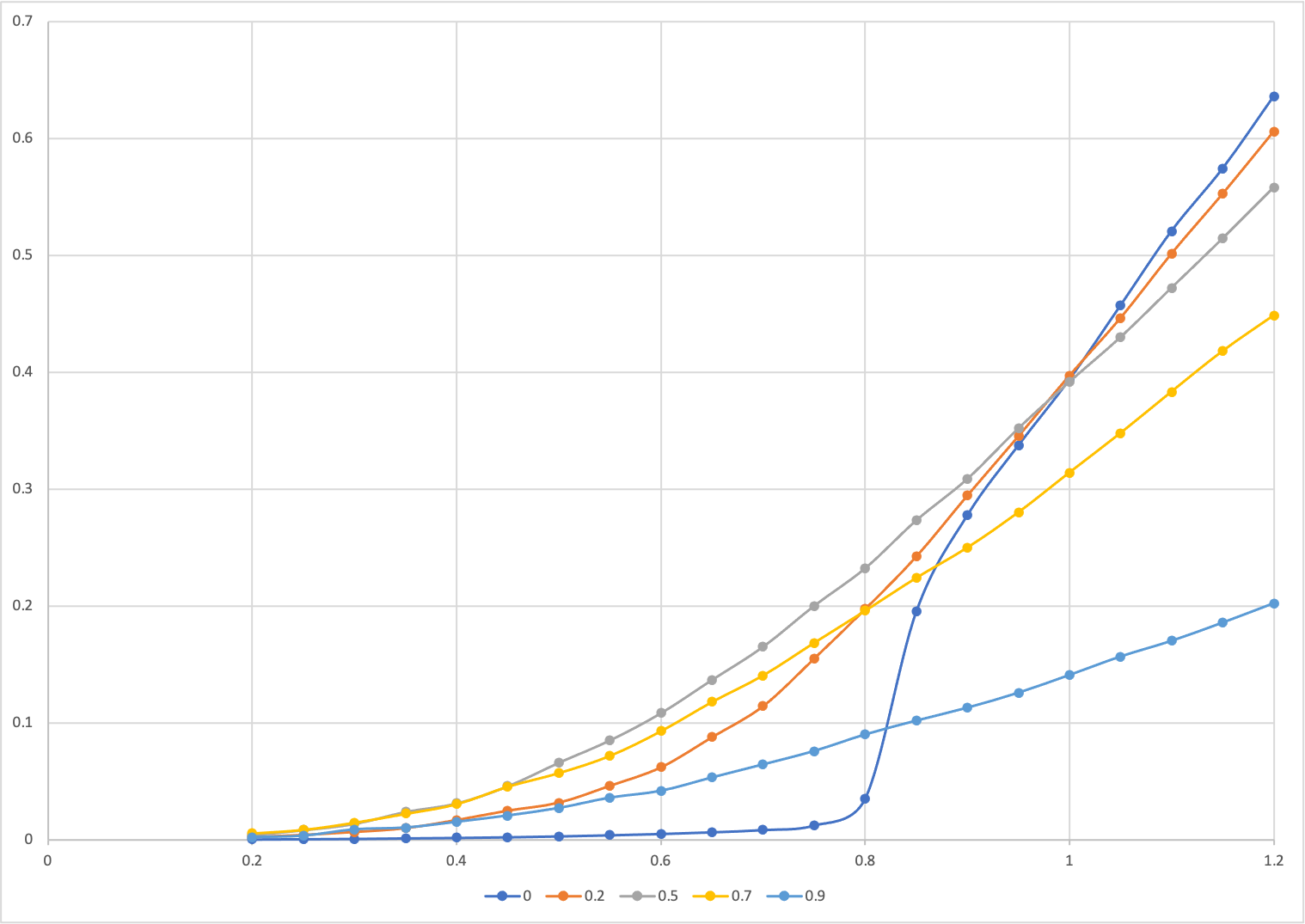}
		\end{center}
	\caption{Average error as a function of $\lambda=m/n$, for Zipf's distributions with $\beta \in \{0.0,0.2,0.5,0.7,0.9\}$. Plots obtained for $n=1000$, $k=3$, $N=50,000$.}
	\label{fig:zipf}
\end{figure}
We observe that the phase transition behavior vanishes for $\beta>0$. It turns out that even in the subcritical regime, frequent elements, while having no error themselves, heavily affect the error of certain rare elements, which raises the resulting average error. In the supercritical regime ($\lambda>1$ in Figure~\ref{fig:zipf}) the opposite happens: the uniform distribution shows the largest average error. This is because an increasingly large fraction of the keys (those in the core of the associated hypergraph) contribute to the error, while for skewed distributions, frequent keys tend to have no error, and thus the larger $\beta$ (with frequent keys becoming more predominant) the smaller the average error.  {Note that this is in accordance with the observation of \cite{DBLP:conf/teletraffic/BianchiDLS12} that the estimates for the uniform distribution majorate the estimates of infrequent keys for skewed distributions.}




\section{Concluding remarks}

We presented an analysis of conservative update strategy for Count-Min sketch under the assumption of uniform distribution of keys in the input stream. Our results show that the behaviour of the sketch heavily depends on the properties of the underlying hash hypergraph. Assuming that hash functions are $k$-wise independent, the error produced by the sketch follows two different regimes depending on the density of the underlying hypergraph, that is the number of distinct keys relative to the size of the sketch. When this ratio is below the threshold, the conservative update strategy produces a $o(1)$ relative error when the input stream and the number of distinct keys both grow, while the regular Count-Min produces a positive constant error. This gap formally demonstrates that conservative update achieves a substantial improvement over regular Count-Min. 

We showed that the above-mentioned threshold corresponds to the peelability threshold for $k$-uniform random hypergraphs. One practical implication of this is that the best memory usage is obtained with three hash functions, due to the fact that $\lambda_3$ is maximum among all $\lambda_k$, and therefore $k=3$ leads to the minimum number of counters needed to deal with a given number of distinct keys. 

In \cite{DBLP:conf/sigmod/CohenM03} it is claimed, without proof, that the rate of positive errors of conservative update is $k$ times smaller than that of regular Count-Min. This claim does not appear to be true. Note that Count-Min does not err on a key represented in the sketch if and only if the corresponding edge of the hypergraph includes a leaf (vertex of degree 1), while the conservative update can return an exact answer even for an edge without leaves. However, this latter event depends on the relative frequencies of keys and therefore on the specific distribution of keys and the input length. On the other hand, our experiments with uniformly distributed keys show that this event is relatively rare, and the rate of positive error for Count-Min and conservative update are essentially the same. 


One important assumption of our analysis is the uniform distribution of keys in the input. 
We presented an experimental evidence that  for skewed distributions, in particular  for Zipf's distribution, the phase transition 
{disappears when the skewness parameter grows. Therefore, the uniform distribution presents the smallest error in the subcritical regime. The situation is the opposite in the supercritical regime when the number of distinct keys is large compared to the number of counters: here the uniform distribution presents the largest average error.
As mentioned earlier, } 
for Zipf's distribution, frequent keys have essentially no error, whereas in the supercritical regime, low frequency keys have all similar overestimates. This reveals another type of phase transition in error approximation for Zipf's distribution, occurring between frequent and infrequent elements. Quantifying this transition is an interesting open question directly related to the accurate detection of \textit{heavy hitters} in streams.

\paragraph*{Acknowledgments} We thank Djamal Bellazougui who first pointed out to us the conservative update strategy. 

\bibliography{biblio.bib}

\end{document}